\documentclass[11pt,letterpaper]{article}

\usepackage{geometry}
\usepackage[utf8]{inputenc}
\usepackage{amsmath, amsthm, amssymb}
\usepackage[vlined, ruled]{algorithm2e}

\usepackage{tikz,pgf}
\usepackage{mathrsfs}
\usetikzlibrary{arrows}

\usetikzlibrary{decorations.pathmorphing,patterns,scopes}

\usepackage{comment}
\usepackage{adjustbox}
\usepackage{xcolor}
\usepackage{hyperref}
\hypersetup{
    colorlinks,
    linkcolor={red!50!black},
    citecolor={blue!50!black},
    urlcolor={blue!80!black}
}
\usepackage{txfonts}
\usepackage{MnSymbol}

\newcommand{\R}{\ensuremath{\mathbb{R}}}

\newcommand{\Disc}[1]{\ensuremath{\mathcal{D}}\left(#1\right)}

\newcommand{\ie}{\textit{i.e.}}

\newcommand{\ang}{\text{ang}}

\usepackage{subcaption}

\newtheorem{theorem}{Theorem}
\newtheorem{lemma}{Lemma}

\begin{document}

\title{Probabilistic Asynchronous Arbitrary Pattern Formation\thanks{This work was performed within the Labex SMART supported by French state funds managed by the ANR within the Investissements d'Avenir programme under reference ANR-11-IDEX-0004-02.}}

\author{Quentin Bramas$^1$ \and Sébastien Tixeuil$^1$}

\date{}


\maketitle
\begin{center}
    $^1$Sorbonne Universités, UPMC Univ Paris 06, CNRS, LIP6 UMR 7606, 4 place Jussieu 75005 Paris.\\
    \texttt{\{quentin.bramas,\,sebastien.tixeuil\}@lip6.fr}
\end{center}

\begin{abstract}
We propose a new probabilistic pattern formation algorithm for oblivious mobile robots that operates in the ASYNC model. Unlike previous work, our algorithm makes no assumptions about the local coordinate systems of robots (the robots do \emph{not} share a common ``North'' nor a common ``Right''), yet it preserves the ability from any initial configuration that contains at least $5$ robots to form any general pattern (and not just patterns that satisfy symmetricity predicates). Our proposal also gets rid of the previous assumption (in the same model) that robots do not pause while moving (so, our robots really are fully asynchronous), and the amount of randomness is kept low -- a single random bit per robot per Look-Compute-Move cycle is used. Our protocol consists in the combination of two phases, a probabilistic leader election phase, and a deterministic pattern formation one. As the deterministic phase does not use chirality, it may be of independent interest in the deterministic context. A noteworthy feature of our algorithm is the ability to form patterns with multiplicity points (except the gathering case due to impossibility results), a new feature in the context of pattern formation that we believe is an important asset of our approach.\\

\end{abstract}


\section{Introduction}

We consider a set of mobile robots that move freely in a continuous 2-dimensional Euclidian space. 
Each robot repeats a Look-Compute-Move (LCM) cycle~\cite{Suzuki99}. First, it \emph{Looks} at its surroundings to obtain a snapshot containing the locations of all robots as points in the plane, with respect to its ego-centered coordinate system. Based on this visual information, the robot \emph{Computes} a destination and then \emph{Moves} towards the destination.
The robots are identical, anonymous and oblivious \ie, the computed destination in each cycle depends only on the snapshot obtained in the current
cycle (and not on the past history of execution). The snapshots obtained by the robots are not consistently oriented in any manner.

The literature defines three different models of execution: in the fully synchronous (FSYNC) model, robots execute LCM cycles in a lock-step manner, in the semi-syn\-chronous (SSYNC) model, each LCM cycle is supposed atomic, 
and in the most general asynchronous (ASYNC) model, each phase of each LCM cycle may take an arbitrary amount of time. This last model enables the possibility that a robot observes another robot while the latter is moving (and moving robots appear in the snapshot \emph{exactly the same way static robots do}), and that move actions are based on obsolete observations.

In this particularly weak model it is interesting to characterize which additional assumptions are necessary and sufficient for the robots to cooperatively perform a given task. In this paper, we consider the pattern formation problem in the most general ASYNC model. The robots start in an arbitrary initial configuration where no two robots occupy the same position, and are given the pattern to be formed as a set of coordinates in their own local coordinate system. An algorithm solves the pattern formation problem if within finite time the robots form the input pattern and remain stationary thereafter.

\noindent\textbf{Related Works.}
The pattern formation problem has been extensively studied in the deterministic setting~\cite{bouzid2011pattern,BouzidDPT10,Yamashita2010,Yamashita2012,Shantanu2010,YamashitaS10,Yamauchi13,DieudonnePV10,Flocchini08,Suzuki99}. 
The seminal paper on mobile robots~\cite{Suzuki99} presents a deterministic solution to construct general patterns in the SSYNC model, with the added assumption that robots have access to an infinite non-volatile memory (that is, robots are \emph{not} oblivious). The construction was later refined for the ASYNC model by Bouzid \emph{et al.}~\cite{bouzid2011pattern}, still using a finite number of infinite precision variables. 

The search for an oblivious solution to the general pattern formation proved difficult~\cite{Flocchini08}. For oblivious deterministic robots to be able to construct any general pattern, it is required that they agree on a common ``North'' (that is, a common direction and orientation) but also on a common ``Right'' (that is, a common chirality), so that robots get to all agree on a common coordinate system. If only a ``North'' (and implicitly if only a ``Right'') is available, then some patterns involving an even number of robots cannot be formed. 
Relaxing the common coordinate system condition let to a characterization of the patterns that can be formed by deterministic oblivious robots~\cite{Yamashita2010,Yamashita2012,Yamauchi13}. The best deterministic algorithm so far in the ASYNC model without a common coordinate system~\cite{Yamashita2012} proves the following: If $\rho$ denotes the geometric symmetricity of a robot configuration (i.e., the maximum integer $\rho$ such that the rotation by $2\pi/\rho$ is invariant for the configuration), and $I$ and $P$ denote the initial and target configurations, respectively, then $P$ can be formed if and only if $\rho(I)$ divides $\rho(F)$. All aforementioned deterministic solutions assume that both the input configurations and the target configuration do not have multiplicity points (that is, locations hosting more than one robot), and that robots share a common chirality. Overall, oblivious deterministic algorithms either need a common coordinate system or cannot form any general pattern. 

To circumvent those impossibility results, the probabilistic path was taken by Yamauchi and Yamashita~\cite{Yamashita2014}. The robots are oblivious, operate in the most general ASYNC model, and can form any general pattern from any general initial configuration (with at least $n\geq5$ robots), without assuming a common coordinate system. However, their approach~\cite{Yamashita2014} makes use of three hypotheses that are not proved to be necessary: \emph{(i)} all robots share a common chirality, \emph{(ii)} a robot may not make an arbitrary long pause while moving (more precisely, it cannot be observed twice at the same position by the same robot in two different Look-Compute-Move cycles while it is moving), and \emph{(iii)} infinitely many random bits are required (a robot requests a point chosen uniformly at random in a continuous segment) anytime access to a random source is performed. While the latter two are of more theoretical interest, the first one is intriguing, as a common chirality was also used extensively in the deterministic case. The following natural open question raises: is a common chirality a necessary requirement for mobile robot \emph{general} pattern formation ? As the answer is yes in the deterministic~\cite{Flocchini08} case, we concentrate on the probabilistic case.

\noindent\textbf{Our contribution.}
In this paper, we propose a new probabilistic pattern formation algorithm for oblivious mobile robots that operate in the ASYNC model. Unlike previous work, our algorithm makes no assumptions about the local coordinate systems of robots (they do \emph{not} share a common ``North'' nor a common ``Right''), yet it preserves the ability from any initial configuration that contains at least $5$ robots to form any general pattern (and not just patterns such that $\rho(I)$ divides $\rho(F)$). Besides relieving the chirality assumption, our proposal also gets rid of the previous assumption~\cite{Yamashita2014} that robots do not pause while moving (so, they really are fully asynchronous), and the amount of randomness is kept low -- a single random bit per robot is used per use of the random source -- (\emph{vs.} infinitely many previously~\cite{Yamashita2014}). Our protocol consists in the combination of several phases, including a deterministic pattern formation one. As the deterministic phase does not use chirality, it may be of independent interest in the deterministic context. 

A noteworthy property of our algorithm is that it permits to form patterns with multiplicity points (\emph{without} assuming robots are endowed with multiplicity detection), a new feature in the context of pattern formation that we believe is an important asset of our approach. Of course, the case of gathering (a special pattern defined by a unique point of multiplicity $n$) remains impossible to solve in our settings~\cite{prencipe2007impossibility}.

%

\section{Model}
Robots operate in a 2-dimensional Euclidian space. Each robot has its own local coordinate system. For simplicity, we assume the existence (unknown from the robots) of a global coordinate system. Whenever it is clear from the context, we manipulate points in this global coordinate system, but each robot only sees the points in its local system. Two set of points $A$ and $B$ are \emph{similar}, denoted $A\approx B$, if $B$ can be obtained from $A$ by translation, scaling, rotation, or symmetry. A \emph{configuration} $P$ is a set of positions of robots at a given time. Each robot that looks at this configuration may see different (but similar) set of points. 

Each time a robot is activated it starts a Look/Compute/Move cycle. After the \emph{look} phase, a robot obtains a configuration $P$ representing the positions of the robots in its local coordinate system. After an arbitrary delay, the robot \emph{computes} a path to a destination. Then, it \emph{moves} toward the destination following the previously computed path. The duration of the move phase, and the delay between two phases, are chosen by an adversary and can be arbitrary long. The adversary decides when robots are activated assuming a \emph{fair} scheduling \ie, in any configuration, all robots are activated within finite time. The adversary also controls the robots movement along their target path and can stop a robot before reaching its destination, but not before traveling at least a distance $\delta > 0$ ($\delta$ being unknown to the robots).

An execution of an algorithm is an infinite sequence $P(0), P(1), \ldots$ of configurations. 
An algorithm $\psi$ forms a pattern $F$ if, for any execution $P(0), P(1), \ldots$, there exists a time $t$ such that $P(t)\approx F$ and $P(t')=P(t)$ for all $t'\geq t$. In the sequel, the set of points $F$ denotes the pattern to form. The coordinates of the points in $F$ are given to the robots in an arbitrary coordinate system so that each robot may receive different, but equivalent, pattern $F$. If the pattern contains points of multiplicity, the robots receives a multiset, that is, a set where each element  is associated with its multiplicity. Even if the robots are not endowed with multiplicity detection, they know from the pattern what are the points of multiplicity to form. In particular, then can deduce from the pattern, the number $n$ of robots, even if they do not see $n$ robots.

\section{Algorithm Overview}
Our algorithm is divided into four phases. Since robots are oblivious and the scheduling is asynchronous, we cannot explicitly concatenate several phases to be executed in a specific order. However, one can simulate the effect of concatenation of two (or more) phases by inferring from the current configuration which phase to execute. Implementing this technique is feasible if phases are associated with disjoint sets of configurations where they are executed. Also, in order to simplify the proof of correctness, robots should not switch phases when placed in a configuration containing moving robots \emph{i.e.}, a phase has to ensure that if the configuration resulting from a movement is associated with another phase, then all the robots are static (that is, none of them is moving). When this property holds, the first time a phase is executed, we can suppose that the configuration is static. 


Our algorithm can form an arbitrary pattern. In particular, the pattern $F$ can contain points of multiplicity (but cannot be a single point). If this is the case, the robots create a new pattern $\tilde{F}$ from $F$ where they remove the multiplicity, and add around each point $p$ of multiplicity $m$, $m-1$ points really close to $p$, and located at the same distance to the center of the smallest circle enclosing $F$. The algorithm then proceed as usual with $\tilde{F}$ instead of $F$. The initial pattern $F$ is formed by the termination phase, when $\tilde{F}$ is almost formed. So from now, we suppose that the pattern does not contains points of multiplicity and we refer to the details of the termination phase to see how an arbitrary pattern is formed.

In the following we define the phases of our algorithm and the set of associated configurations, starting from the more precise one (the phase we intuitively execute at the end to complete the pattern formation). Unless otherwise stated, the \emph{center} of a configuration refers to the center of the smallest enclosing circle of this configuration.
\vspace{0.2cm}

\noindent\textbf{Termination.} The termination phase occurs when all robots, except the closest to the center, forms the target pattern (from which we remove one of the point closest to the center). The phase consists in moving the last robot towards its destination. While moving, the robot remains the closest to the center, so that the resulting configuration is associated to the same phase.
\vspace{0.2cm}

\noindent\textbf{Almost Pattern Formation.} Among the remaining configurations, we associate the guided ones to this phase. A configuration is \emph{guided} when a unique robot is sufficiently close to the center and induces by its position a global sense of direction and orientation to every robot. In particular, when executing this phase, robots are totally ordered and have a unique destination assigned. The phase consists first in moving all the robots (except the one that is closest to the center), one by one, so that they are at the same distance to the center as their destination in the pattern. Secondly, the robots moves toward their destination, keeping their distance to the center and the ordering unchanged. The configuration has to remain guided until each robot, except the closest to the center, reaches its destination. A configuration obtained after executing this phase is either associated to the same phase, or to the termination phase.

\noindent\textbf{Formation of a Guided Configuration 1 (FBC1).} Among the remaining configurations, we associate the ones that contain a \emph{centered equiangular or biangular} (CEB) set to this phase. A CEB-set is a subset of robots that exists when the configuration is symmetric or is almost symmetric. Moreover it is constructed independently from the coordinate system (so it is unique when it exists), and is invariant when the robots in this set move toward (or away from) the center of the configuration (see Section~\ref{phase2} for a formal definition). When the configuration contains a CEB-set, our algorithm consists in moving the robots in this set to obtain a guided configuration. The invariance property of the CEB-set is important to ensure that resulting configurations are still associated with this phase.

In more details, when this phase is executed, the robots in the CEB-set $Q$ moves either toward or away from the center with probability $1/2$. We show that, with probability $1$, a unique robot is elected after a finite number of activations. Then, the elected robot performs a special move to force the other robots in $Q$ to terminate their movement. Once each robot is static, the elected robot moves toward the center to create a guided configuration. During the execution of this phase, it is possible that the configuration is associated with the termination phase. If this happens, our algorithm makes sure that all robots are static.

\noindent\textbf{Formation of a Guided Configuration 2 (FBC2).} We associate all remaining configurations to this phase. When executing this phase, the configuration does not have a CEB-set. This implies that the configuration is not symmetric, so the robots are totally ordered. Therefore, the smallest robot moves toward the center to create a guided configuration (it remains the smallest robot while doing so). Before the movement, the robot checks if there exists a point in its path that creates a configuration containing a CEB-set. If it is the case, the robot chooses this point as its destination so that, when the configuration contains a CEB-set (and the robots switch to the FBC1 phase), all the robots are static. 

\newcommand\phaseTransitionSize\scriptsize

\tikzstyle{every node}=[]
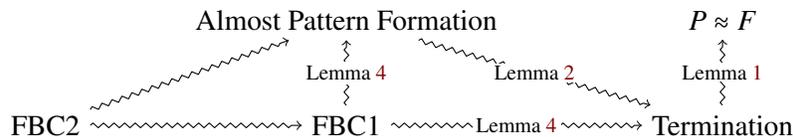
\begin{figure}[h]\centering
    \begin{tikzpicture}
    \node (DPF) at (0,0) {Almost Pattern Formation};
    \node [below of=DPF, node distance=1.4cm] (RSB-Q)  {FBC1};
    \node [right of=RSB-Q, node distance=5cm] (Term) {Termination};
    \node [left of=RSB-Q, node distance=4cm] (RSB-Qb)  {FBC2};
    \node [above of=Term, node distance=1.4cm] (P-F)  {$P\approx F$};
    \draw [->,
        line join=round,
        decorate, decoration={
            zigzag,
            segment length=4,
            amplitude=.9,post=lineto,
            post length=2pt
        }]  (RSB-Q) -- (DPF) node [midway, fill=white,inner sep=0.05cm] {\phaseTransitionSize Lemma \ref{lem:phase FCB1 terminates}};
    \draw [->,
        line join=round,
        decorate, decoration={
            zigzag,
            segment length=4,
            amplitude=.9,post=lineto,
            post length=2pt
        }]  (RSB-Qb) -- (RSB-Q) node [midway] {}
        ;
    \draw [->,
        line join=round,
        decorate, decoration={
            zigzag,
            segment length=4,
            amplitude=.9,post=lineto,
            post length=2pt
        }]  (RSB-Qb) -- (DPF) node [midway] {}
        ;
    \draw [->,
        line join=round,
        decorate, decoration={
            zigzag,
            segment length=4,
            amplitude=.9,post=lineto,
            post length=2pt
        }]  (RSB-Q) -- (Term) node [midway, fill=white,inner sep=0.05cm] {\phaseTransitionSize Lemma \ref{lem:phase FCB1 terminates}};
    \draw [->,
        line join=round,
        decorate, decoration={
            zigzag,
            segment length=4,
            amplitude=.9,post=lineto,
            post length=2pt
        }]  (DPF) -- (Term) node [midway, fill=white,inner sep=0.03cm] {\phaseTransitionSize Lemma \ref{lem:phase 3 terminates and C(P) is unchanged}};
        \draw [->,
        line join=round,
        decorate, decoration={
            zigzag,
            segment length=4,
            amplitude=.9,post=lineto,
            post length=2pt
        }]  (Term) -- (P-F) node [midway, fill=white,inner sep=0.03cm] {\phaseTransitionSize Lemma \ref{lem:Termination terminates}};
\end{tikzpicture}
\caption{Relations between the phases of the algorithm.}\label{fig: relation between phases}
\end{figure}
We define the relation $\rightsquigarrow$ between phases, where $A\rightsquigarrow B$ if executing the phase $A$ can lead to a configuration associated with phase $B$. 
The proof of our main theorem follows from Lemma~\ref{lem:Termination terminates}, Lemma~\ref{lem:phase 3 terminates and C(P) is unchanged}, and Lemma~\ref{lem:phase FCB1 terminates} 
(see Figure~\ref{fig: relation between phases}) detailed in the next Section.

\begin{theorem}
    Our algorithm forms any pattern $F$ that is not a point, starting from any configuration of at least 5 robots that does not contain a point of multiplicity.
\end{theorem}

\section{Algorithm Details}
In this section we describe in more detail the phases of our algorithm. 
We start by listing the necessary notations used in the remaining of the section.

\subsection{Notations}

Let $P$ be a set of points, then $C(P)$ denotes the smallest enclosing circle of $P$. Otherwise mentioned, $c(P)$ denotes the center of $C(P)$. The circle of a robot $r\in P$ is the circle centered at $c(P)$ containing $r$. We say a robot moves on its circle if its trajectory is contained in its circle. A radial movement is a linear movement whose origin and destination are on the same half-line of origin $c(P)$. We say a robot moves radially if it performs a radial movement.

The pattern to form, $F$, is given to each robot as a set of points in their local coordinate system. However, at each activation, robots can translate and scale their local coordinate system so that $C(P)=C(F)$. Hence, we suppose in the remainder of the paper that $C(P)=C(F)$, and that the radius of $C(P)$ is the common unit distance (unless otherwise mentioned). This is possible because in our case, the configuration where all robots share the same location (that is, are gathered) is not reachable. For two points $a$ and $b$, $|a|_b = |a-b|$ denotes the distance between $a$ and $b$. In a $n$-robot configuration $P$, as we are often interested in the distance between a point and the center $c(P)$, we simply write $|a|$ instead of $|a|_{c(P)}$. 

The interior, resp. the exterior, of a disc or a circle $C$, denoted $interior(C)$, resp. $exterior(C)$, does not include the circumference. A set of points $A$ (or simply a point) \emph{holds} $C(P)$ if $C(P\setminus B) \neq C(P)$, for a subset $B\subset A$.
The angle formed by three points $u$, $v$ and $w$ is denoted by $ang(u,v,w)\in [0,2\pi)$, and the orientation depends on the context. If the orientation is not given, it is either clockwise or counterclockwise, but it does not vary for a given robot during a cycle.

\vspace{0.2cm}

\noindent\textbf{Partial Ordering and Symmetricity.}
Given a set of points (typically a configuration) $P$, we order the points based on their coordinates in the coordinate systems defined by the points that are the closest to, but not at, the center $c(P)$. Formally, let $M=\{r\in P\text{ \emph{s.t.} } |r|=\min_{r'\in P\setminus\{c(P)\}}|r'|\}$. For each robot $r_m$ in $M$ and each orientation $o$ in $\{\lcirclearrowright, \rcirclearrowleft\}$ we define the polar coordinate system $Z_{r_m}^o: P\rightarrow \R^+\times[0,2\pi)$, $r\mapsto (|r|, ang(r_m, c(P), r))$ oriented by $o$. 
We denote by $Z_{r_m}^o(P)$ the increasing sequence of coordinates in $Z_{r_m}^o$ of the points in $P$. In particular, if $r_1$ and $r_2$ are two points such that $|r_1| < |r_2|$, then the coordinates in $Z_{r_m}^o(P)$ of $r_1$ are smaller than the coordinates of $r_2$. 

We define $Z_{\min}$ as the set of coordinate systems that minimize the sequence of coordinates, using the lexicographical order. We use this set of coordinate systems to define the relation $<$, where $r < r'$ if and only if the coordinates of $r$ are smaller than the coordinates of $r'$ in every coordinate systems $Z \in Z_{\min}$. From this relation we deduce the \emph{partial ordering of $P$}. 

We define the \emph{symmetricity} of a configuration $P$ as the number of minimal points in its partial ordering. In particular, if the robots are endowed with chirality (or if the configuration does not have an axis of symmetry), this definition matches the definition of symmetricity of previous work~\cite{BouzidDPT10,Yamashita2012}. In particular, if a configuration $P$ is such that $\rho(P) = k > 1$, then $P$ can be partitioned in $n/k$ regular $k$-gons centered at $c(P)$ (where a $2$-gon is a line with center its middle). However, this is not true in the general case when robots do not have a common sense of chirality. Also, it is important to notice that if $c(P) \in P$ then $\rho(P) = 1$, even if the configuration is symmetric or is invariant by rotation.
\vspace{0.2cm}

\noindent\textbf{Ordered and Guided Configuration.}
An \emph{ordered} configuration is a configuration where the partial order of robots is a \emph{total} order. In particular, this implies that there is a unique coordinate system $Z$ in $Z_{\min}$ and all the robots agree on $Z$ as a global coordinate system.

When the configuration is ordered, let $f_1, f_2, \ldots, f_n$ be any total ordering of points in $F$ satisfying the partial ordering of $F$ (an arbitrary ordering can be chosen if more than one satisfies the condition). Even if the pattern is given to robots using an arbitrary coordinate system, each robot can scale it so that $C(F) = C(P)$, mirrors it so that the orientation chosen for the ordering of the points in $F$ coincides with the orientation of the ordering of the robots, and rotate it so that the points $f_2$ and $r_2$ are on the same half-line of origin $c(P)$. Without loss of generality, we can suppose that $F$ is given to the robots with those property, in the global coordinate system $Z$.

One can observe that the choice of the ordering of points in $F$ is not important, since the resulting coordinates of points in $F$ in the global coordinate system $Z$ are identical for two different orderings (indeed, the resulting sets are equivalent, and are equals after applying the aforementioned transformations). So, from now on, when the configuration is ordered, robots see the points in $F$ in the same way in the global coordinate system, and have a common ordering of points in $F$.

A \emph{guided configuration} is an ordered configuration that satisfies: 
    \emph{(i)} $|r_1| = |r_{2}|/2$; 
    \emph{(ii)} $|r_{2}| \leq |f_{2}|$;
    \emph{(iii)} $2ang(r_1,c(P), r_{2}) < \min_{f\neq f_1,\,|f|=|f_2|} ang(f_2, c(F), f)$ (see Figure~\ref{fig:example of a guided configuration}).

A static configuration that only satisfies $|r_1| \leq \min(|r_2|, |f_2|)/2$ can be easily converted to a static guided configuration. To do so, $r_2$ moves toward the center until $|r_{2}| \leq |f_{2}|$ and if the third condition is not satisfied, $r_1$ moves to the center and then moves away from it to form an appropriate angle with $r_2$, and such that $|r_1| = \min(|r_2|, |f_2|)/2$. Once $r_1$ and $r_2$ reach their destinations, the configuration is static and guided. Therefore, from now on, a static configuration such that $|r_1| \leq \min(|r_2|, |f_2|)/2$ is considered to be guided.

\subsection{Termination}\label{phase4}

The termination phase consists in the following steps. First, the robots check (a) if a point of multiplicity exists in the configuration (by comparing the number of visible robots with the number of points in the pattern) or (b) if the configuration is guided (so that each robot sees $F$ in the same coordinate system), if each point in $F$ is occupied by a robot, and if the other robots are close to a point of $F$.\\
\emph{If it is not the case}, then the pattern to form is modified to obtain a pattern $\tilde{F}$ without point of multiplicity and the other phases of the algorithm are executed with $\tilde{F}$ instead of $F$.\\
\emph{If it is the case}, then every robot $r$ knows its destination $f_r\in F$ (its the closest point of $F$). If there is a robot $r\neq r_1$ with $r\neq f_r$ and $f_r\neq c(F)$, then $r$ moves toward $f_r$, while remaining in its circle. When no such robot exists then there is two case:\\
\emph{Case 1:} $r_1$ is the only robot not located on its destination in $F$. Then $r_1$ moves toward its destination (if multiple destinations are possible, then it chooses the closest one) and the other robots do not move. While it is moving, the global coordinate system is modified but $r_1$ remains the only closest robots to the center. When $r_1$ reaches its destination (which can be a location that is already occupied by another robot), the pattern $F$ is formed.\\
\emph{Case 2:} There are several robots with destination $c(F)$. Then, they all move toward $c(F)$ when activated. During those moves, the global coordinate system may not be visible to the robots anymore, but all the robots can detect that $F\setminus\{c(F)\}$ is formed and that the robots closest to $c(F)$ are moving toward $c(F)$.

Now we detail how $\tilde{F}$ is constructed from $F$. The goal is that, when $\tilde{F}$ is formed by the robots, every point of $F$ is occupied by a robot and the other robots are close to a point of $F$ (so that they know their destination in $F$ even without the ordering of robot). Let $F$ be the initial pattern that can contain points of multiplicity. Let $\tilde{F}$ be the set of points constructed from $F$ by removing the multiplicity and adding, for each point $p\in F\setminus\{c(F)\}$ of multiplicity $m>1$, $m-1$ points $p_1, \ldots, p_{m-1}$ such that $|p_i| = |p|$ and $ang(p, c(F), p_i) = \frac{\alpha(F)}{4i}$, with $\alpha(F)=\min \left(\{ ang(f,c(F),f') \;|\; f,\,f'\in F\} \setminus \{0\}\right)$. If $c(F)$ has multiplicity $m$ in $F$ then, we also add $m-1$ points evenly distributed on the circle of radius $\varepsilon_F = \frac{1}{4}\min_{f\in F\setminus\{c(F)\}} |f|$ such that one point is on the half line of origin $c(F)$ passing through one of the greatest point in $F$. 
In this construction, the orientation of the angles is either deduced from the pattern, or arbitrary if $F$ has an axis of symmetry. The construction may not be unique (e.g. if the $F$ has an axis of symmetry) but the possible resulting sets are all similar, and $\tilde{F}$ is any of those. 
For all the phases except Termination, $\tilde{F}$ is used instead of $F$ whenever $F$ contains points of multiplicity.

\begin{lemma}\label{lem:Termination terminates}
    The \emph{Termination} phase terminates on a configuration $P\approx F$.
\end{lemma}

\subsection{Almost Pattern Formation}\label{phase3}
In this phase we assume that the configuration is guided. We have a total order over robots $r_1, \ldots,r_n$, and each robot sees the pattern in the same way in a global coordinate system $Z$. In $Z$, the points of $F$ are also totally ordered $f_1,\ldots,f_n$ so that each robot $r_i$ knows its final destination $f_i$. The goal of this phase is that each robot, except $r_1$, reaches its destination. In the sequel $P'=P\setminus\{r_1\}$ and $F'=F\setminus\{f_1\}$. One can observe that $C(P') = C(P)$, and we can assume that $C(F') = C(F)$ (if this is not the case, we can modify the ordering of points in $F$ so that $f_1$ does not hold $C(F)$).

This phase consist in two sub-phases. The first one consists in moving each robot to its correct circle and the second one in moving each robot to its destination, while remaining in its circle. While doing this, the ordering of robot stays the same so that the destination of each robot remain unchanged. We give here a detailed overview, and refer to the appendix for a complete description.
 
\noindent\textbf{Reach the Correct Circle.}
Let $C_1$, $C_2$, $\ldots$, $C_m$ be the $m$ circles centered at $c(P)$ with decreasing radius, each containing at least one point in $F'$. This subphase consists in moving robots so that each circle contains the correct number of robots. For each circle $C_i$, $1\leq i\leq m$, we remove a robot if there are too many robots on $C_i$ and we add a robot if there are too few robots on $C_i$. We do so by either moving a robot that is on the circle toward the inside of $C_i$ or the contrary. We move the robots one by one while keeping the ordering of robot and by keeping $C(P)$ unchanged.

\noindent\textbf{Reach the Destination.}
Let $i\in[1,m]$, $C_i$ now contains $m_i$ robots and the $m_i$ destinations for those robots. The robots and the destinations are ordered so that each robot is aware of its corresponding destination. They can all move toward their destination, while remaining on $C_i$ and preserving the robots ordering (\ie, without reaching another robot position). When a robot $r$ is active and another robot is on the way, $r$ chooses on the circle half the distance to this robot. There cannot be a deadlock since there is no cycle in the waiting relation. Indeed, robots on $C_i$ are ordered by angle so that they behave like they are on a finite segment. 
If $i=1$, during their movement, robots also ensure that $C(P)$ remains unchanged. To do so, if a robot $r\in C_1$ is active and detects that its movement can modify $C(P)$, then it moves as much as possible without changing $C(P)$.

\begin{lemma}\label{lem:phase 3 terminates and C(P) is unchanged}
    After the \emph{Almost Pattern Formation} phase, each robot in $P'$ reaches its destination in $F'$, and the \emph{Termination} phase is executed.
\end{lemma}

\begin{proof}
First we show that there is no deadlock. Suppose we have on a circle $C$, $m$ robots $r_1 < r_2 < \ldots < r_m$ and $m$ destinations $d_1 < d_2 < \ldots <  d_m$. $r_i$ has destination $d_i$ and moves toward it (staying on $C$) in the direct orientation if $r_i < d_i$ and in the indirect orientation otherwise. For the sake of contradiction, suppose that $r_i < d_i$ and $r_i$ cannot reach $d_i$, even after an infinite number of activation. we observe that, to block $r_i$, $r_{i+1}$ must satisfy $r_i< r_{i+1} \leq d_1$ and $r_{i+1}$ is not able to free the way for $r_{i+1}$. This implies that $r_{i+1}\leq r_{i+3}\leq d_1$. Recursively, this means that $r_m \leq d_1$. But nothing blocks $r_m$ to reach $d_m$ when $r_m\leq d_m$, a contradiction.

Now suppose for the purpose of contradiction that $C(P)$ is modified. This means that there exist two robots $r$ and $r'$ on $C_1$ that form an angle greater than $\pi$. 
Before $C(P)$ is modified, they form an angle of at most $\pi$, so that one robot's movement on $C_1$ in the direct orientation, and the other's movement on $C_1$ in the indirect orientation. This is possible only if there is no point in $F$ on $C_1$ between $r$ and $r'$, which is a contradiction with the fact that $C(P)=C(F)$.
\end{proof}


\subsection{Formation of a Guided Configuration 1}\label{phase2}

If the current configuration $P$ is not guided, and contains a CEB-set $Q$, we execute this phase to obtain a guided configuration. To do so, one robot has to be elected to guide the configuration. Once a unique robot is elected, other robots may still be moving. One way to be sure that the other robots are static is to give them new destinations. To do so, the elected robot moves on its circle with a small angle. After this move, the other robots can still detect the elected robot and compute the angle. This angle can be used as a persistent memory. For a given angle, each robot moves toward a deterministic destination. This ensure that at the end of the phase all the robots are static. Another angle is used at the end of the phase to form a guided configuration.

The phase consists of two procedures: a robot election and using the elected robot to form a static guided configuration.

First, we give the formal definition of the CEB-set and how to compute it. Then, we present the two procedures that use the property of the CEB-set to form a static guided configuration.

\vspace{0.2cm}

\noindent\textbf{CEB-set Definition.}
For a given point $c$, the \emph{string of angles} $SA_{c, r, o}(P)$ starting from a robot $r$ with orientation $o$ is the sequence of angles formed by the robots in $P$ with $r$, around Point $c$, and with Orientation $o$. If for two robots $r\neq r'$, the strings of angles $SA_{c, r, o}(P)$ and $SA_{c, r', o}(P)$ are equal, we say the configuration is \emph{regular} (this is tantamount to say the string of angles is periodic). Regular sets have been introduced by Bouzid et al.~\cite{BouzidDPT10} for solving the gathering problem. Its main property is that there exists at most one point $c$ such that $SA_{c, r, o}(P)=SA_{c, r', o}(P)$, and this point is invariant by robots movement toward to, or away from $c$ (as it is also the Weber point). Particular regular configurations are equiangular configurations (the period of the string of angles is $1$) and biangular configurations (the period of the string of angles is $2$). The point $c$ is called the center of regularity of the configuration.

We say that a configuration is \emph{sym-regular}, if two strings of angles, centered at the center $c(P)$ of $C(P)$, with opposite orientations, are equal, \emph{i.e.}, if there exists $r,r'\in P$ (possibly $r=r'$) such that $SA_{c(P), r, \lcirclearrowright}(P)=SA_{c(P), r', \rcirclearrowleft}(P)$. 

We define the \emph{centered equiangular of biangular set} (CEB) set $CEB(P)$ of a configuration $P$ as follow: \\
\noindent $\bullet$ {If $P$ is not regular nor sym-regular:} then $P$ does not have a CEB-set \ie, $CEB(P) = \emptyset$.
\\
\noindent $\bullet$ {If the whole configuration $P$ is equiangular or biangular:} then $CEB(P) = P$ and, \emph{in this case, the center of the configuration $c(P)$ is the center of regularity}. \\
\noindent $\bullet$ {Otherwise:} the CEB-set $Q$ is constructed as follow. Initially, $Q=\emptyset$.
Let $S$ be the set of smallest robots (according to the partial ordering of robots) such that $Q\cup S$ does not hold $C(P)$. If $Q\cup S$ is equiangular or biangular then add the robots of $S$ in $Q$ and start again, otherwise stop and $CEB(P) = Q$.
Informally, the CEB set is the biggest subset of $P$ containing smallest robots that is equiangular or biangular with center $c(P)$ and that does not hold $C(P)$. Since $n\geq 5$, if all robots are on $SEC(P)$, there exist some robots that can be moved without changing $SEC(P)$, so that $Q\neq \emptyset$.


\begin{algorithm}[h]\small
    \SetAlgoRefName{constructCEBSet}
    \caption{}
    \label{algo:constructCEBSet}
    $Q = \emptyset$\quad $Ignore\leftarrow\emptyset$\\
    \While{$P\setminus (Ignore\cup Q)\neq\emptyset$}{
    Let $S$ be the set of smallest robots in $P\setminus (Ignore\cup Q)$ according to the partial ordering of robot\\
    \textbf{if} $Q\cup S$ holds $P$ \textbf{then} $Ignore\leftarrow Ignore\cup S$\\
    \textbf{else if }$Q\cup S$ is equiangular or biangular \textbf{then} $Q\leftarrow Q\cup S$\\
    \textbf{else} stop
    }
\end{algorithm}


\begin{theorem}\label{thm:sym conf have a CEB}
    For a configuration $P$, if $\rho(P)>1$, then $CEB(P)\neq\emptyset$.
\end{theorem}
If by moving one robot, a configuration has a CEB-set and the robot is among the closest robots to the center of the CEB-set, we say that the configuration has a CEB-set with a \emph{shifted robot}. The shifted robot is the robot that we need to move to create the CEB-set. The configuration is seen as if the shifted robot is at its right position (\ie, the position where it has to be for the configuration to contain a CEB-set). The difference between the angles of this robot before and after the move (with origin $c(P)$) is called the \emph{shift angle}. 
\begin{theorem}\label{thm:uniqueness of the regular set}
    Let $n\geq5$, and $P$ be an $n$-robot configuration that contains a CEB-set with a shifted robot. Let $\theta(P)$ be the smallest (non-null) angle, centered at $c(P)$, between two robots in the configuration.
    If the shift angle is at most $\theta(P)/2$, then the shifted robot is unique.
\end{theorem}
The proof of this result is not trivial, and can be found in the appendix.
When computing the CEB-set of a configuration, all the robots can see if the CEB-set has a shifted robot. The existence of a shifted robot is crucial for our algorithm because the shift angle can be used as a persistent memory. In more details, if all the other robots move radially, they do not change the shift angle, so that the shifted robot can modify the shift angle to remember some information for the next activation. Moreover, all the other robots can see the shift angle and deduce some information from it.
In our algorithm the shift angle is used to order the robots to make a specific move, and then, another shift angle is used to stop them.
\vspace{0.2cm}

\noindent\textbf{Robot Election.} Let $Q$ bet the CEB-set. 
 We say a robot $r_e$ is \emph{elected} if it is the shifted robot of $Q$ or if $|r_e| < \frac{7}{8}\min_{r\in Q\setminus\{r_e\}}|r|$.
To elect a robot, each robot $r$ in $Q$ proceeds in the following way. If there is another robot in $Q$ that is strictly closer to the center, then $r$ does not move. 
If $r$ is not elected and is one of the closest robot (unique or not), then $r$ chooses randomly (each choice with probability $1/2$) to go toward or away from the center $c(P)$. 
If $r$ chooses to move toward the center, it moves a distance $|r|/8$. 
If $r$ chooses to move away to the center, it moves a (possible null) distance $\min\left(\frac{1}{2}\left(d-\left|r\right|\right), \frac{1}{7}|r|\right)$, where $d$ is the minimum distance to the center among robots in $P\setminus Q$ (and $d=\infty$ if $P\setminus Q=\emptyset$). This ensures that robots in $Q$ remain in the CEB-set in the resulting configuration.
A robot is aware that it is elected if it is elected during its look phase.
When a robot is aware it is elected, and is not yet the shifted robot, it moves on its circle to create a shift angle of $ \theta/8$. After the robot election, some robots may still be moving.
\begin{lemma}\label{lem:robot election} The following properties hold: \emph{(i)} eventually one robot is aware it is elected with probability one, and \emph{(ii)} once a robot is aware it is elected, another robot cannot be elected.
\end{lemma}

Also, our algorithm has to ensure that $n-1$ robots cannot form part of the pattern. To do so, a procedure is called making sure that, if a point in $F$ is in the path of a robot, then this point is either chosen as its destination, or avoided (the procedure is described in the appendix).
So that, if the configuration is associated with the \emph{Termination} phase, all the robots are static. 
\vspace{0.2cm}

\noindent\textbf{Using the Shifted Robot to Form a Static Guided Configuration.} 
If the elected robot is not shifted or if the shift angle is in $[0,\theta/8)$ (with $\theta = \theta(P)$), it moves on its circle to create a shift $\theta/8$. Also, if the shift angle is in $(\theta/8, \theta/4)$ and if the robots in the CEB-set $Q$ are not on the same circle as the shifted robot, then the shifted robot moves on its circle to create a shift $\theta/8$. If another robot is activated during the movement of the shifted robot (\ie, when the shift is not exactly $\theta/8$), it chooses not to move. 
When the shift angle is exactly $\theta/8$, the elected robot waits for the other robots in $Q$ to reach its circle. When this is the case, the shifted robot moves on its circle to create a shift angle of $\theta/4$. Then it moves toward the center to create an guided configuration \ie, the shifted robot $r_1$ moves radially such that $|r_1| = \min(|r_2|, |f_2|)/2$. 


\begin{lemma}\label{lem:phase FCB1 terminates}
    After executing the \emph{FCB1} phase, the configuration is static and associated with either the \emph{Termination} phase or the \emph{Almost Pattern Formation} (in finite time with probability one).
\end{lemma}

\section{Concluding Remarks}

Similarly to previous work~\cite{Yamashita2014}, the initial configuration should not contain multiplicity positions. 
In the case where the initial configuration contains points of multiplicity, a convenient solution would be to reuse known pattern formation algorithms (such as ours) and run a preliminary phase where multiplicity points are eliminated. This task is known as the scattering task in the literature~\cite{Bramas2015}. However, even the most recent developments~\cite{Bramas2015} only considers the SSYNC model. Of course, as our protocol also performs correctly in SSYNC, it is possible to combine the two to obtain a protocol in SSYNC that manages multiplicities both in $I$ and in $F$. Indeed, combining protocols in SSYNC is facilitated because moves are always aware of the latest configuration, so for all configurations that have multiplicities and do not belong to a legitimate path toward the target pattern, the scattering phase is run, until robots either reach a configuration where there is no point of multiplicity or a configuration that makes progress toward the target pattern. Extending this scheme to the ASYNC model requires to solve the open problem of ASYNC scattering, and making sure the combinations of protocols is feasible.


\bibliographystyle{plain}
\bibliography{async_pattern_formation}

\newpage

\appendix

\section{Pseudo-code}

The algorithm consists of several phases (that are not exactly divided in the same way as in the paper), described in pseudo-code in algorithm \ref{formPattern}. Each procedure call has a phase condition. A phase is executed if and only if its phase condition is not verified. If the condition is verified, the next phase is considered. Each time a robot is activated, it must find the first phase with a condition that is not verified and follow the corresponding instructions. Each phase is done not to break the previous phase conditions. The condition line \ref{line:last movement condition} is checked before because the movement line \ref{line:last movement} breaks the condition of the other phases.
Line $6$ corresponds to the robot election. The goal is to select a robot by performing random radial movement in the CEB-set. Lines \ref{line:createGuidedConfiguration} to \ref{line:rotateRobotOnCircle} correspond to the deterministic pattern formation algorithm. For simplicity, our pseudo-code does not handle the case of a pattern that contains points of multiplicity.

{\LinesNumbered
\begin{algorithm}[]\scriptsize
    \SetAlgoRefName{formPattern}
    \caption{main algorithm that forms a pattern F}
    \label{formPattern}
    \If{$m = |c(F)| > 1$} {
        
        \eIf{$F\setminus\{c(F)\}$ is formed} {
            the $m$ robots closest to the center move toward $c(F)$    \\
            \textbf{Return}    
        } {
            $F\Leftarrow F\setminus\{c(F)\}\cup\{\text{$m$ distinct points close to $c(F)$}\}$   
        }
    }
    $ClosestF \leftarrow \{\text{smallest elements in $F$ that does not hold $C(F)$}\}$\\
    $ClosestP \leftarrow \{\text{smallest elements in $P$}\}$\\
    \eIf{$ClosestP = \{r\}$ and $\exists f\in ClosestF$ s.t. $P\setminus\{r\}\approx F-\{f\}$\label{line:last movement condition}}{
          $r$ moves toward the closest $f$\label{line:last movement}
    }{
        $r_1\leftarrow selectARobot()$\label{line:selectARobot}\\    
    }
    
    \eIf{each point of $F$ is occupied except maybe one of the smallest and the other robots are close to their destinations} {
        Each robot $r\neq r_1$ moves toward the closest point in $F$  \\
        \textbf{Return}   
    }{
        $F\leftarrow$ $F$ where every points of multiplicity $m$ is replaces by $m$ points close to it and on the same circle.
    }

    $P'\leftarrow P\setminus\{r_1\}$\label{line:createGuidedConfiguration}\\
    $F'\leftarrow F\setminus\{f_1\}$ (with $f_1\in ClosestF$)\\
    
    Let $C_1$, $C_2$, $\ldots$, $C_m$ be the $m$ circles centered at $c(P)$ with decreasing radius, each containing at least one point in $F'$. For each $1\leq i\leq m$, let $m_i = |C_i\cap F'| > 0$\\
    \If{$|C_i\cap F'|=2$}
    {
        $fixEnclosingCircle()$\\
    }
    \For{$i=1, 2,\ldots, m$}{
        $cleanExterior(i)$\label{line:cleanExterior}\\
        $locateEnoughRobots(i)$\label{line:locateEnoughRobots}\\
        $removeRobotsInExcess(i)$\label{line:removeRobotsInExcess}
    }
    $rotateRobotOnCircle()$\label{line:rotateRobotOnCircle}
    
\end{algorithm}}

\begin{algorithm}\scriptsize
    \SetAlgoRefName{handlePartiallyFormedPattern}
    \caption{executed before the robot election to handle configuration that can create a configuration that verifies in line \ref{line:last movement condition} of the main algorithm}
    \label{handlePartiallyFormedPattern}
    \If{$\{F_r\cap F_r^c\}$ is a partition of $F$ such that $F_r^c\approx P\setminus Q$ \textbf{and}\\ $|Q| - 1$ robots in $Q$ are located on an halfline $[c(P), f)$, with $f\in F_r$}{
        $d_1\leftarrow$ radius of $C(F_r)$\\
        $d_2\leftarrow \min\{d\;|\;\Disc{d_1}\cap exterior(\Disc{d}) \cap F_r = \emptyset\}$\\
        $d\leftarrow (d_1+d_2)/2$\\
        \If{$\exists r\in Q$ s.t. $|r|>d$}{
            \eIf{$\exists r\in Q$ s.t. $|r|>d_1$}{
                \For{$r\in Q$ s.t. $|r|>d_1$}{
                    $r$ moves radially at distance $|d_1|$ from $c(P)$
                }
            }{
                \For{$r\in Q$ s.t. $|r|>d$}{
                    $r$ moves radially at distance $|d|$ from $c(P)$
                }
            }
            \textbf{exit}
        }
    }
\end{algorithm}

\begin{algorithm}[]\scriptsize
    \SetAlgoRefName{selectARobot}
    \caption{select a robot}
    \label{selectARobot}
    \textbf{Phase Condition:} There exists a selected robot $r_s$\\
    \textbf{Returned Value:} $r_s$\\
    \uIf{$P$ contains a CEB-set $Q$ with shifted robot}
        {
            $r_e\leftarrow$ the shifted robot\\
            $\varepsilon\leftarrow$ the shift angle\\
            $S\leftarrow \{r\in P \;|\; |r|>|r_e|\}$\\
            \uIf{$S\neq\emptyset$ and $\varepsilon\neq \theta/8$}{
                $r_e$ moves on its circle to create a $\theta/8$-shifted angle          
            }
            \uElseIf{$S\neq\emptyset$ and $\varepsilon=\theta/8$}{
                  \For{$r\in S$}{
                        $r$ moves radially at distance $|r|$ from $c(P)$                  
                  }          
            }
            \uElseIf{$\varepsilon<\theta/4$}{

                $r_e$ moves on its circle to create a $\theta/4$-shifted angle        
            }
            \Else{
                $r_e$ moves radially toward $c(P)$ to become selected
            }
        }
    \uElseIf{$P$ contains a CEB-set $Q$}
        {
            
            \eIf{$P\setminus Q \neq \emptyset$}{
                $d\leftarrow\min_{r'\in P\setminus Q}|r'|$
            }{
                $d\leftarrow\infty$
            }
            handlePartiallyFormedPattern()\\
            \For{$r\in P$}{
                \uIf{$|r| < \frac{7}{8}\min_{r\in Q\setminus\{r\}}|r|$}
                {
                    $r$ moves on its circle to create a $\theta/8$-shifted angle
                }
                \ElseIf{$\{r'\neq r\;\textrm{ s.t. }\;|r'|<|r|\} = \emptyset$}{
                    $c\leftarrow$ 1 with probability $1/2$, 0 otherwise\\
                    \eIf{$c$}{
                        $r$ moves a distance $|r|/8$ toward $c(P)$
                    }{
                        $r$ moves a distance $\min\left(\frac{1}{2}\left(d-\left|r\right|\right), \frac{1}{7}|r|\right)$ away from $c(P)$
                    }
                }
             } 
        }\uElse
        {
            $r_{1}\leftarrow$ unique robot with maximum view that does not hold $C(P)$\\
            \eIf{$\exists r\in [r_{1}, c(P)]$, $P\cup\{r\}\setminus\{r_{1}\}$ is biangular}{
                $r_{1}$ moves toward $r$
            }{
                $r_{1}$ moves toward $c(P)$
            }
        }
\end{algorithm}



\begin{algorithm}[]\scriptsize
    \SetAlgoRefName{cleanExterior(i)}
    \caption{remove robots outside $C_i$}
    \label{cleanExterior(i)}
    \textbf{Phase Condition:} $i = 1$ or $|interior(C_{i-1}\cap exterior(C_i)\cap P'| = 0$\\
    $r \leftarrow$ smallest robot in $exterior(C_i)$\\
    $C\leftarrow$ circle centered at $c(P)$ that contains $r$\\
    \eIf{$|C\cap P| > 1$}{
        $r$ moves toward $c(P)$ without reaching the circle of another robot nor $C_i$
    }{
        $a\leftarrow \max_{r' \in C_i}ang(r_{2}, c(P), r')$   
        \eIf{$ang(r_{2}, c(P), r) > a$} {
            $r$ moves toward $c(P)$ to reach $C_i$        
        }{
            $r$ moves on $C_i$ in the direct orientation to have an angle $(2\pi + a)/2$
        }
    }
\end{algorithm}

\begin{algorithm}[]\scriptsize
    \SetAlgoRefName{locateEnoughRobots(i)}
    \caption{locate enough robots on $C_i$}
    \label{locateEnoughRobots(i)}
    \textbf{Phase Condition:} $|C_{i}\cap P'| \geq m_i$\\
    $r \leftarrow$ greatest robot in $interior(C_i)$\\
    $C\leftarrow$ circle centered at $c(P)$ that contains $r$\\
    \eIf{$|C\cap P| > 1$}{
        $r$ moves away from toward $c(P)$ without reaching the circle of another robot nor $C_i$
    }{
        $a\leftarrow \min_{r' \in C_i}ang(r_{2}, c(P), r')$  \\ 
        \eIf{$ang(r_{2}, c(P), r) < a$} {
            $r$ moves away from $c(P)$ to reach $C_i$        
        }{
            $r$ moves on $C_i$ in the indirect orientation to have an angle $a/2$
        }
    }
\end{algorithm}

\begin{algorithm}[]\scriptsize
    \SetAlgoRefName{removeRobotsInExcess(i)}
    \caption{remove robot in excess on $C_i$}
    \label{removeRobotsInExcess(i)}
    \textbf{Phase Condition:} $|C_{i}\cap P'| = m_i$\\
        // \textit{Where $Poly(a,b)$ denotes the set of vertice of the regular $a$-gon centered at $c(P)$ that have the line $c(P)r_{2}$ as axis of symmetry union $b$ points evenly distributed in the arc between angle $0$ and $\pi/a$}\\
    \eIf{$i > 1$}
    {
        $r \leftarrow $ smallest robot on $C_i$\\
        $r$ moves toward $c(P)$ without reaching the circle of another robot    
    }{
        \eIf{robots the $m_1$ greatest robots on $C_1$ forms $Poly(m_1, 0)$}{
            $r \leftarrow $ smallest robot on $C_1$\\
            $r$ moves toward $c(P)$ without reaching the circle of another robot 
        }{
            robots on $C_1$ form $Poly(m_1, |P\cap C_1| - m_1)$
        }
    }
\end{algorithm}

\begin{algorithm}[]\scriptsize
    \SetAlgoRefName{rotateRobotOnCircle}
    \caption{move the robots on their circle to reach their final destination}
    \label{rotateRobotOnCircle}
    \textbf{Phase Condition:} $F'=P'$\\
    Let $r_1, \ldots, r_{n-1}$ be the robots in $P'$ in the lexicographic order of their polar coordinates in the global coordinate system.\\
    Let $d_1, \ldots, d_{n-1}$ be the point of $F'$ in the lexicographical order of their polar coordinates in the global coordinate system.\\
    
    \For{$i = 1, \ldots, n-1$}{
        $A\leftarrow$ the arc of the circle of $r_i$ delimited by $r_i$ and $d_i$ that does not contains the point of angle $0$\\
        \If{$A\cap P'\neq \emptyset$}{
            $c\leftarrow$ closest robot in $A\cap P'$\\
            $d\leftarrow$ point of $A$ in the middle $r_i$ and $c$\\
            $A\leftarrow$ the arc of the circle of $r_i$ delimited by $r_i$ and $d$
        }
        \If{$r_i\in C(P)$ and }{
            $d\leftarrow$ the farthest point on $A$ so that $C(P)$ does not change 
        }    
        $r_i$ moves on $A$ toward $d$
    }
\end{algorithm}

\section{Missing Algorithm Details}

\subsection{\texorpdfstring
{Pattern Formation when $|C(F)\cap F'|=2$}
{Pattern Formation when |C(F) inter F'|=2}}\label{sec:pattern formation when C(F) inter F' = 2}

\begin{algorithm}[]\scriptsize
    \SetAlgoRefName{fixEnclosingCircle}
    \caption{locate the robot of $C(P)\cap P'$ on $C(P)\cap F'$ when $|C(F)\cap F'|=2$}
    \label{fixEnclosingCircle}
    \textbf{Phase Condition:} $|C(F)\cap F'|\neq 2$ or there are only two robot in $C(P)$ located on the two point of $C(F)\cap F'$\\
    \eIf{$|C(P)\cap P'| = 2$}{
        $r\leftarrow$ greatest robot in interior $C(P)$\\
        $C\leftarrow$ circle centered at $c(P)$ that contains $r$\\
        \eIf{$|C\cap P| > 1$}{
            $r$ moves away from toward $c(P)$ without reaching the circle of another robot nor $C(P)$
        }{
            $a\leftarrow \min_{r' \in C(P)}ang(r_{2}, c(P), r')$   \\
            \eIf{$ang(r_{2}, c(P), r) < a$} {
                $r$ moves away from $c(P)$ to reach $C(P)$        
            }{
                $r$ moves on $C(P)$ in the indirect orientation to have an angle $a/2$
            }
        }
    }{
        $r\leftarrow$ greatest robot in $C(P)$\\
        $r'\leftarrow$ smallest robot in $C(P)$\\
        \eIf{$r$ and $r'$ are located the points of $C(P)\cap F'$}
        {
            $r''\leftarrow$ second smallest robot in $C(P)\cap P'$\\
            $r''$ moves toward $c(P)$ without reaching the circle of another robot
        }{
            Let $r_1, \ldots, r_k$ be the other robots in $C(P)\cap P'$ in the lexicographical order of their polar coordinates\\
            \emph{// perform the following movements while preserving $C(P)$ and the ordering of robots}\\
            $r$ moves on $C(P)$ toward the greatest point in $C(P)\cap F'$\\
            $r'$ moves on $C(P)$  toward the smallest point in $C(P)\cap F'$\\
            \For{$i=1,\ldots,j$}
            {
                $a\leftarrow ang(r_{2}, c(P), r') + i\times(ang(r_{2}, c(P), r')+ang(r_{2}, c(P), r))/(j+1)$\\
                $r_i$ moves on $C(P)$ toward the point in $C(P)\cap F'$ with angle $a$
            }
        }
    }
\end{algorithm}

We execute this special phase before executing the first sub-phase of phase~\ref{phase3}, if $|C(F)\cap F'|=2$. If $|C(F)\cap F|=2$ and there are not exactly two robots in $C(P)$ located at the two points in $C(F)\cap F$, then the following is executed.

If there are only two robots on $C(P)$, then the greatest robot in $interior(C(P))$ reaches $C(P)$, while remaining smaller than robots in $C(P)$ (see Action~$locateEnoughRobots(i)$). Now, there are at least three robots on $C(P)$. The greatest robot $r$ in $C(P)$ moves toward the greatest robot in $C(F)$, the smallest $r'$ moves toward the other point in $C(P)\cap F'$, and the other robots choose evenly distributed destinations between $r$ and $r'$. Those movements are done while keeping $C(P)$ and the ordering unchanged. The smallest robot is chosen for $r'$ instead of the second greatest so that no robot can prevent $r'$ to reach the smallest point in $C(P)\cap F'$, especially if it has a null angle. Once $r$ and $r'$ reach their destination, the other robots can leave $C(P)$, starting from the smallest. Those last movements change the ordering of $r$, so that it becomes the second greatest robot.

\subsection{Almost Pattern Formation}

\noindent\textbf{Reach the Correct Circle.}
This sub-phase consists in moving robots so that there is the right number of robots on each circle centered at $c(P)$.
Let $C_1$, $C_2$, $\ldots$, $C_m$ be the $m$ circles centered at $c(P)$ with decreasing radius, each containing at least one point in $F'$. For each $1\leq i\leq m$, let $m_i = |C_i\cap F'| > 0$. We have $\sum_{i=1}^m m_i = |F'| = n-1$.

Before beginning this sub-phase, the robots that have a null angle (except $r_2$) move on their circle following the direct orientation while preserving the order (\ie, without reaching another robot), so that no robot has a null angle (except $r_2$). This is required for proper operation of action $ii)$, defined below. Also, if $m_1=2$, since two robots cannot move on $C(P)$ synchronously to keep $C(P)$ unchanged, we need to execute a special procedure (see subsection~\ref{sec:pattern formation when C(F) inter F' = 2}) to ensure that the two robots are located at the two points of $C(P)\cap F'$, keeping $C(P)$ unchanged. Informally, this procedure moves another robot on $C(P)$ if there are only two robots on it, then the two greatest robots reach their destination point in $C(P)\cap F'$. Then, the other robots can leave safely $C(P)$. From now on, we suppose that if $m_1=2$, then $C_1$ already contains two robots located at their corresponding point in $F$.

Recursively, we move robots such that each circle $C_i$ contains exactly $m_i$ robots.
We define the following procedure for a given $i$, $1\leq i\leq m$. The procedure executes three actions sequentially and assumes, if $i>1$, that $|interior(C_{i-1}) \cap P'| = \sum_{j=i} ^{m}m_j$.
\tolerance=3000
\begin{itemize}
    \item[\emph{i)}]
    \emph{\textbf{cleanExterior(i):}} 
    If $i>1$, moves robots so that $interior(C_{i-1}) \cap exterior(C_i) \cap P'{=}\emptyset$: the smallest robot in $exterior(C_i)$ moves to $C_i$ while it remains greater than robots already in $C_i$. To do so, it can moves a little toward $c(P)$, so that there is no other robot in its circle, then moves on its circle so that its angle is greater than the angles of robots in $C_i$, and finally moves radially toward $c(P)$ to reach $C_i$ (\emph{e.g.}, movement of robot $r_5$ in Figure~\ref{fig:execution of action i and ii}). If $i=m$, we also ensure that its angle is less than $2\pi - ang(r_1, c(P), r_2)$, in order not to break the guided configuration. We repeat this procedure until there are no more robots between $C_{i-1}$ and $C_i$ 
    (\emph{e.g.}, movement of robots $r_5$ and $r_6$ in Figure~\ref{fig:execution of action i and ii}). 

    \item[\emph{ii)}]
    \emph{\textbf{locateEnoughRobots(i):}} Move robots so that $|C_i \cap P'| \geq m_i$: we have $|interior(C_{i}) \cap P'| \geq 1$. Indeed, if $i = 1$, $|interior(C_{1}) \cap P'| = m_1-|C_1 \cap P'|\geq 1$, otherwise, there are by hypothesis at least $m_i$ robots inside $C_{i-1}$, and after performing action $i)$, theses robots are not between $C_i$ and $C_{i-1}$. The greatest robot in $interior(C_i)$ now moves to $C_i$ while remaining smaller than robots already in $C_i$. To do so, it can move a little away from $c(P)$ so that no other robot remains in its circle, then move on its circle so that its angle is smaller than the angles of robots in $C_i$ (but not null), and  finally move radially away from $c(P)$ to reach $C_i$ 
    (\emph{e.g.}, movement of robot $r_7$ in Figure~\ref{fig:execution of action i and ii}).

\item[\emph{iii)}]
    \emph{\textbf{removeRobotsInExcess(i)}} Move robots so that $|C_i \cap P'|\leq m_i$: the smallest robot in $C_i$ moves a little toward $c(P)$ (here, ``a little'' means a small distance such that the order is preserved, \ie, the robot does not reach the circle of another robot nor $C_{i+1}$). We repeat this process until there are exactly $m_i$ robots on $C_i$.\\
    If $i=1$, make sure that $C(P)$ does not change. (see Figures \ref{fig:remove robots in excess when i=1} and \ref{fig:remove robots in excess when i=1, bis}). 
    However, we know that $m_1\geq 3$.    The $m_1$ greatest robots $r_{n}, \ldots, r_{n-m_1}$ remain on $C_1$, and have to be the only robots to hold $C(P)$. To do so, the angles formed by two consecutive robots in $\{r_{n}, \ldots, r_{n-m_1}\}$ have to be smaller than, or equal to $\pi$. This is obtained by moving the robots on $C_1$, while preserving the ordering and $C(P)$, such that $r_{n}, \ldots, r_{n-m_1}$ form the regular $m_1$-gon that has the line $\overline{c(P)r_{2}}$ as an axis of symmetry (see Figure~\ref{fig:remove robots in excess when i=1}). At the same time, if the $m_1$-gon is not formed yet, other robots in $C_1$ move on $C_1$ to be evenly distributed in the arc between angle $0$ and $\pi/m_1$ (see the blue arc in Figure~\ref{fig:remove robots in excess when i=1}), again while preserving the ordering and $C(P)$. Overall, each robot on $C_1$ has a deterministic (and non-blocking) destination. Once the $m_1$-gon is formed (even if some other robots are still moving), the smallest robot in $C(P)\cap P$ moves a little toward $c(P)$ (see Figure\ref{fig:remove robots in excess when i=1, bis}). This is repeated until only $r_{n}, \ldots, r_{n-m_1}$ remain on $C_1$.
\end{itemize}
After executing those actions for $i=1,2, \ldots, m$, each circle contains the proper number of robots \ie,  there are $m_i$ robots on $C_i$.


\tikzstyle{every node}=[circle, draw, fill=black!50,%
                        inner sep=0pt, minimum width=4pt]

\tikzstyle{fakeNode}=[circle, draw=black, fill=white]
\begin{figure}[htb]\centering
    \begin{subfigure}[t]{0.32\textwidth}\centering
       \begin{tikzpicture}[every path/.style={>=stealth}, scale=0.80]
        
        \draw[dashed] (180:1.5) arc (180:-180:1.5);
        
        \draw[dashed] (0:0) -- (0:1.5);
        \draw[<-] (0:1.6) -- (0:3);
        \draw[dashed] (0:0) -- (-60:2.3);
        \draw[dashed] (0:0) -- (-20:3);
        \draw[<->] (0:1.9) arc (0:-60:1.9);
        
        \draw   (-40:2.1) node[opacity=0, text opacity=1] {$\theta$};        
        
        \draw[<->] (0:2.2) arc (0:-20:2.2);
        \draw   (-10:3) node[opacity=0, text opacity=1] {$<\theta/2$};        
        
        \draw (0:0) node[minimum width=1.5pt,fill=black, label={180:$c(P)$}] {};
        \draw (-20:0.4) node[label={-70:$r_{1}$}] {};
        \draw (0:1.2) node[label={110:$r_{2}$}] {};
        \draw (0:1.5) node[fill=white, label={30:$f_{2}$}] {};
        \draw (-60:1.5) node[fill=white] {};

        \end{tikzpicture}
        \caption{\small part of a guided configuration, located near the center.}\label{fig:example of a guided configuration}
    \end{subfigure}
~~~\begin{subfigure}[t]{0.30\textwidth}\centering
        \begin{tikzpicture}[every path/.style={>=stealth}, scale=0.55]
        
        \draw (180:3.3) node[opacity=0, text opacity=0] {};
        \draw (0:4) node[opacity=0, text opacity=0] {};

        \draw (30:1.8) node[opacity=0, text opacity=1] {$C_2$};
        \draw[dashed] (0:0) circle (1.5);
        \draw (40:3.3) node[opacity=0, text opacity=1] {$C_1$};
        \draw[dashed] (0:0) circle (3);
        \draw[] (0:3.3) -- (180:0);
        \draw[dashed] (0:0) -- (-20:1.9);
        
        \draw[color=red, line width=0.6mm] (0:1.5) arc (0:-40:1.5);
        \draw[color=red] (-40:1.5) node[fill=red] {};
        \draw (0:1.5) node[label={170:$r_{2}$}] {};
        \draw (-20:0.75) node[label={190:$r_{1}$}] {};
        
        \draw (160:1.5) node[label={-45:$r_{3}$}] {};

        \draw (-90:1.5) node[label={-135:$r_{4}$}] {};

        \draw (-130:2.3) -- (-130:2.0);
        \draw[] (-130:2.0) arc (-130:-80:2.0);
        \draw[->] (-80:2.0) -- (-80:1.5);
        \draw (-130:2.3) node[label={-90:$r_{5}$}] {};
        
        \draw[] (-110:2.3) arc (-110:-60:2.3);
        \draw[->] (-60:2.3) -- (-60:1.5);
        \draw (-110:2.3) node[label={-90:$r_{6}$}] {};
        
        \draw[] (110:2.7) arc (110:60:2.7);
        \draw[->] (60:2.7) -- (60:3);
        \draw (110:2.7) node[label={-45:$r_{7}$}]  {};
        
        \draw (90:3) node[label={45:$r_{8}$}] {};
        \draw (180+30:3) node[label={180:$r_{9}$}] {};
        \draw (-90+30:3) node[label={-45:$r_{10}$}] {};

        \end{tikzpicture}
        \caption{\small Execution of functions $locateEnoughRobots(i)$ and $cleanExterior(i)$}\label{fig:execution of action i and ii}
    \end{subfigure}
    ~~~\begin{subfigure}[t]{0.30\textwidth}\centering
        \begin{tikzpicture}[every path/.style={>=stealth}, scale=0.55]

        \draw (30:1.8) node[opacity=0, text opacity=1] {$C_2$};
        \draw[dashed] (0:0) circle (1.5);
        \draw (40:3.3) node[opacity=0, text opacity=1] {$C_1$};
        \draw[dashed] (0:0) circle (3);
        \draw[] (0:3.3) -- (180:0);
        \draw[dashed] (0:0) -- (-20:1.9);
        
        \draw[color=red, line width=0.6mm] (0:1.5) arc (0:-40:1.5);
        \draw[color=red] (-40:1.5) node[fill=red] {};
        
        \draw (0:1.5) node[fill=white, inner sep = 0.08cm] {};
        \draw (0:1.5) node[label={170:$r_{2}$}] {};
        \draw (-20:0.75) node[label={190:$r_{1}$}] {};
        
        \draw (160:1.5) node {};
        \draw (-90:1.5) node {};
        \draw (-80:1.5) node {};
        \draw (-60:1.5) node {};

        \draw[->] (160:1.65) arc (160:90:1.65);
        \draw[->] (270:1.75) arc (270:110:1.75);
        \draw[->] (280:1.65) arc (280:200:1.65);
        \draw[->] (300:1.70) arc (300:310:1.70);
        
        \draw (90:1.5) node[fill=white] {};
        \draw (110:1.5) node[fill=white] {};
        \draw (-160:1.5) node[fill=white] {};
        \draw (-50:1.5) node[fill=white] {};
        
        \draw (60:3) node {};
        \draw (90:3) node {};
        \draw (180+30:3) node {};
        \draw (-90+30:3) node {};
        
        \draw[->] (60:3.25) arc (60:110:3.25);
        \draw[->] (90:3.15) arc (90:140:3.15);
        \draw (110:3) node[fill=white] {};
        \draw (140:3) node[fill=white] {};
        \draw[->] (210:3.15) arc (210:190:3.15);
        \draw[->] (300:3.15) arc (300:320:3.15);
        \draw (190:3) node[fill=white] {};
        \draw (320:3) node[fill=white] {};
        \end{tikzpicture}
        \caption{\small Execution of function $rotateRobotOnCircle$}\label{fig:reaching destinations}
    \end{subfigure}

\label{fig:Second phase of the pattern formation agorithm}

\begin{subfigure}[b]{0.48\textwidth}\centering
        \begin{tikzpicture}[every path/.style={>=stealth}, scale=0.7]

        \draw (100:1.8) node[opacity=0, text opacity=1] {$C_2$};
        \draw[dashed] (0:0) circle (1.5);
        \draw (80:2.6) node[opacity=0, text opacity=1] {$C_1$};
        \draw[dashed] (0:0) circle (3);
        \draw[] (0:3.3) -- (180:0);
        \draw[dashed] (0:0) -- (-20:1.9);
        
        \draw[color=red, line width=0.6mm] (0:1.5) arc (0:-40:1.5);
        \draw[color=red] (-40:1.5) node[fill=red] {};
        \draw (0:1.5) node[label={170:$r_{2}$}] {};
        \draw (-20:0.75) node[label={190:$r_{1}$}] {};

        \draw[color=blue, line width=0.2mm] (0:3) arc (0:360/6:3);
        
        \draw[dashed] (0:0) -- (360/3-360/6:3.2); 
        \draw[dashed] (0:0) -- (2*360/3-360/6:3.2);   
        \draw[dashed] (0:0) -- (3*360/3-360/6:3.2);   
        \draw[dashed] (360/6/3:2.7) -- (360/6/3:3.4);    
        \draw[dashed] (2*360/6/3:2.7) -- (2*360/6/3:3.3);      
        
        \draw (110:3) node[label={-90:$r_{6}$}] {};
        \draw (140:3) node[label={-45:$r_{7}$}]  {};
        \draw (250:3) node[label={45:$r_{8}$}] {};
        \draw (270:3) node[label={90:$r_{9}$}] {};
        \draw (330:3) node[label={135:$r_{10}$}] {};
                
        \draw[->] (110:3.4) arc (110:360/6/3:3.4);
        \draw[->] (140:3.3) arc (140:2*360/6/3:3.3);
        \draw[->] (250:3.2) arc (250:360/6:3.2);
        \draw[->] (270:3.3) arc (270:2*360/3-360/6:3.3);
        \draw[->] (330:3.2) arc (330:3*360/3-360/6:3.2);
        
        \draw (180:4) node[opacity=0, text opacity=0] {};
        \draw (0:5.5) node[opacity=0, text opacity=0] {};
        
        \end{tikzpicture}
    \caption{Execution of $removeRobotsInExcess(i)$ when $i=1$: the 3 greatest robots form the regular $3$-gon, and the other robots are evenly distributed on blue arc.}\label{fig:remove robots in excess when i=1}
    \end{subfigure}

    ~~~\begin{subfigure}[b]{0.48\textwidth}\centering
        \begin{tikzpicture}[every path/.style={>=stealth}, scale=0.7]

        \draw (180:4) node[opacity=0, text opacity=0] {};
        \draw (0:5.5) node[opacity=0, text opacity=0] {};
        
        \draw (100:1.8) node[opacity=0, text opacity=1] {$C_2$};
        \draw[dashed] (0:0) circle (1.5);
        \draw (80:2.6) node[opacity=0, text opacity=1] {$C_1$};
        \draw[dashed] (0:0) circle (3);
        \draw[] (0:3.3) -- (180:0);
        \draw[dashed] (0:0) -- (-20:1.9);
        
        \draw[color=red, line width=0.6mm] (0:1.5) arc (0:-40:1.5);
        \draw[color=red] (-40:1.5) node[fill=red] {};
        \draw (0:1.5) node[label={170:$r_{2}$}] {};
        \draw (-20:0.75) node[label={190:$r_{1}$}] {};
        
        \draw[color=blue, line width=0.2mm] (0:3) arc (0:360/6:3);

        \draw[dashed] (0:0) -- (360/3-360/6:3.2); 
        \draw[dashed] (0:0) -- (2*360/3-360/6:3.2);   
        \draw[dashed] (0:0) -- (3*360/3-360/6:3.0);   
        
        \draw[->] (2*360/6/3+7:3) -- (2*360/6/3+7:2.5);    
        \draw (360/3/3:2.75) node[opacity=0, text opacity=1] {\tiny$(2)$};    
        \draw[->] (360/6/3:3) -- (360/6/3:2.2);   
        \draw (360/6/3-5:2.7) node[opacity=0, text opacity=1] {\tiny$(1)$};

        \draw (360/6/3:3) node[label={0:$r_{6}$}] {};
        \draw (2*360/6/3+7:3)node[label={0:$r_{7}$}]  {};
        \draw (360/3-360/6:3) node[label={45:$r_{8}$}] {};
        \draw (2*360/3-360/6:3) node[label={35:$r_{9}$}] {};
        \draw (3*360/3-360/6:3) node[label={-90:$r_{10}$}] {};

        \end{tikzpicture}
        \caption{Execution of $removeRobotsInExcess(i)$ when $i=1$: when the 3 greatest robots form the regular $3$-gon, $r_6$ moves toward $c(P)$. Then $r_7$ moves a little toward $c(P)$.}\label{fig:remove robots in excess when i=1, bis}
    \end{subfigure}
\caption{Illustration of the pattern formation algorithm}\label{fig:resolution of the case iv)}
\end{figure}
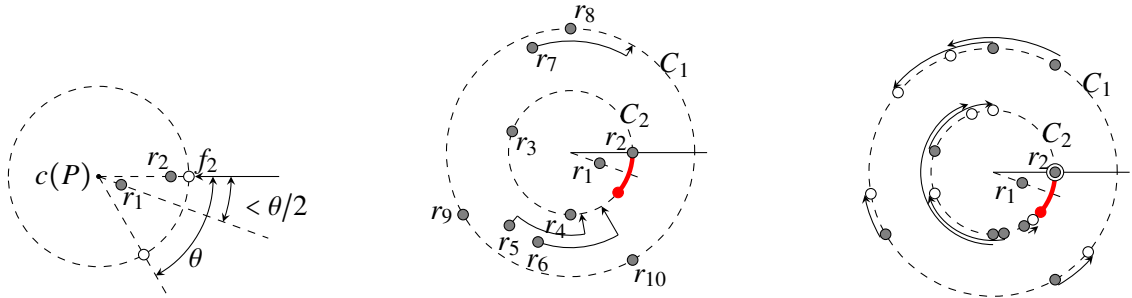
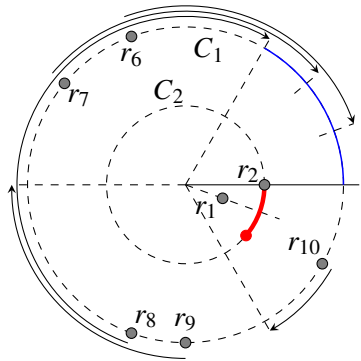
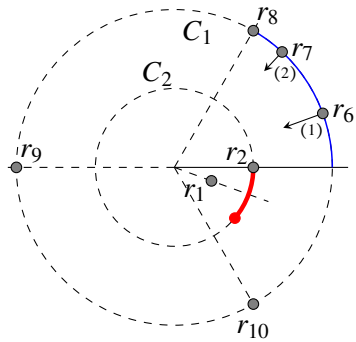

\subsection{Robot Election Pre-phase}

This pre-phase is executed before the robot election, \ie, when the current configuration $P$ contains a CEB-set $Q$ of cardinal $m$.
Before executing the robot election algorithm, a robot checks if the current configuration satisfies the following conditions:
\begin{itemize}
    \item[i)] the pattern can be rotated so that robots in $P\setminus Q$ are located at points in $F$,
    \item[ii)] among the $m$ remaining points of $F$, denoted $F_r$, at least $m-1$ are on $m-1$ half lines, each containing exactly one robot in $Q$.
\end{itemize}
If those conditions are not both satisfied, the robot election is performed as previously described. 
Otherwise three cases can happen. Let $d_1$ be the radius of the smallest circle enclosing $F_r$. If $\Disc{d_1} \cap F_r \neq \emptyset$, let $d_2$ be the smallest radius such that $\Disc{d_1}\cap exterior(\Disc{d_2}) \cap F_r = \emptyset$, otherwise let $d_2 = d_1$. Also, let $d = (d_1 + d_2) / 2$ ($\Disc{a}$ denotes the open disc centered at $c(P)$ of radius $a$)

In the first case, at least one robot $r$ satisfies $|r| > d_1$. Then, all such robots move toward $c(P)$ to reach the circle of radius $d_1$. After each robot reaches its destination, either the whole configuration forms $F$, or $P$ from which we remove the robot with maximum view form $F$ from which we remove a point with maximal view (the configuration is associated with the termination phase), or the configuration is still with the same phase and satisfies the second or the third case. If the configuration is no more associated with the same phase, then the configuration is static.

In the second case, at least one robot $r$ satisfies $|d_1| \geq |r|$ and $|r| > |d|$. Then, all such robots move toward $c(P)$ to reach the circle of radius $d$. After each robot reaches its destination, the configuration satisfies the third case. During this phase the configuration is associated with the same phase since no robot reaches a point in $F_r$.

In the third case, the robots in $Q$ are at most at distance $d$ from $c(P)$. Then the robot election proceeds as previously described, except that a robot with  destination $p$ such that $|p|\geq d$ does not move. During this phase the configuration is associated with the same phase. Indeed, if $d_1 \neq d$, then there is at least one point in $F_r$ that does not contain a robot (and it is not a point with maximum view since some points are closer to $c(P)$), otherwise (all points in $F_r$ are at distance $d$ to $c(P)$), there are at least two robots inside $\Disc{d}$ (because the whole configuration associated with this phase) and they cannot reach the circle of radius $d$.

\subsection{Formation of a Guided Configuration 2}

If $P$ does not have a CEB-set, then, by Theorem~\ref{thm:sym conf have a CEB}, $\rho(P) = 1$ and $P$ has a unique smallest robot $r_1$. 
The robots check if there is a position in the segment $[r_1, c(P))$ such that the whole configuration is equiangular or biangular (possibly with a shifted robot). If it is the case, the first such point in the path becomes the destination of $r_1$. Otherwise, $r_1$ is ordered to move toward the center until the configuration becomes guided. The resulting configurations cannot have a CEB-set (except if it is the whole configuration) since $P$ does not contain one and $r_1$ performs a radial movement.
Starting from a static initial configuration, once the configuration is guided, or once $CEB(P)=P$, all robots are static. Thus we have the following Lemma.
\begin{lemma}\label{lem:phase FCB2 terminates}
    After executing the \emph{FCB2} phase, the configuration is static and associated with either the \emph{FCB1} phase or the \emph{Almost Pattern Formation}.
\end{lemma}

\section{Omitted Proofs}\label{sec:omitted proofs}

\begingroup
\def\thelemma{\ref{lem:Termination terminates}}
\begin{lemma}[restated]
    The \emph{Termination} phase terminates on a configuration $P\approx F$.
\end{lemma}
\addtocounter{lemma}{-1}
\endgroup

\begin{proof}

The phase is executed when the configuration is ordered and every point of $F$ is occupied. If $F$ contains points of multiplicity, and if a robot $r$, distinct from $r_1$, is not located at a point of $F$, then $r$ is close to a point of $F$ \ie, there is a unique point $f\in F$ in its circle such that $ang(f,c(P), r)\leq \alpha(F)/4$. In this case, $r$ moves toward $f$ while remaining on its circle. So that in finite time $r_1$ is the only robot not located at a point of $F$. 

No we assume that $r_1$ is the only robot not located on its destination in $F$. Let $P' = P\setminus\{r_1\}$ and for a multiset $A$, let $U(A)$ the set obtained from $A$ by removing the multiplicity information. The configuration formed by the \emph{visible} robots $U(P')$ may be symmetric even if $P'$ is not. So there may be multiple possible destinations $p$ for $r_1$ such that $U(P')\cup \{p\}$ is similar to $U(F)$ (one can observe that all the possible destination are all at the same distance to the center). We have to show that the closest possible destination is the only remaining point in $F$ where a robot is missing \ie, one of the smallest point of $F$ in the partial ordering. If the smallest point of $F$ is a point of multiplicity, then $r_2$ is already located on it and so, from the definition of a guided configuration (the configuration was guided before a point of multiplicity appeared), $r_2$ is the unique robot that minimize the distance to the center of the configuration and such that 
\begin{equation}\label{eq:apx:condition 3 for a guided configuration}
    2ang(r_1,c(P), r_{2}) < \min_{f\neq f_1,\,|f|=|f_2|} ang(f_2, c(F), f)
\end{equation}
So that $r_1$ can locate $r_2$ and move toward it. Otherwise, if the smallest point $f_1$ of $F$ is isolated, $r_1$ can still detect the position of $r_2 = f_2$. Moreover, by definition of the partial ordering, $ang(f_1, c(P), f_2)$ is minimal so that the destination of $r_1$ is the destination closest to $r_2$. Also, since $f_1$ and $r_1$ are on the same side of the half line $[c(P), r_2)$ (because their positions relative to the half line depend on the orientation of the guided configuration) and since all the possible destinations for $r_1$ are on the same circle, then the closest destination to $r_2$ is also the closest destination to $r_1$ (see Figure~\ref{fig:possible location of f1}). While moving toward its destination, $r_1$ still satisfy condition (\ref{eq:apx:condition 3 for a guided configuration}) so that $r_1$ still detect $r_2$ in the next activation.

 \begin{figure}\centering
       \begin{tikzpicture}[every path/.style={>=stealth}, scale=1.1]
 \begin{scope}
\clip (0:1.5) arc (0:-40:1.5) -- (0:0);
\fill [blue!40] (-5,-5) rectangle (5,5);
\end{scope} 
 \begin{scope}
\clip (240:1.5) arc (240:40:1.5) -- (0:0);
\fill [red!40] (-5,-5) rectangle (5,5);
\end{scope} 

         \begin{scope}
        \draw[dashed] (180:1.5) arc (180:-180:1.5);
        
        \draw[dashed] (0:0) -- (0:1.5);
        \draw[<-] (0:1.6) -- (0:3);
        \draw[dashed] (0:0) -- (-80:2.3);
        \draw[dashed] (0:0) -- (-40:3);
        \draw[<->] (0:1.9) arc (0:-80:1.9);
        \draw   (-60:2.1) node[opacity=0, text opacity=1] {$\theta$};        
        
         
        \draw[<->] (0:2.1) arc (0:-40:2.1);
        \draw   (-30:2.7) node[opacity=0, text opacity=1] {$\theta/2$};        
        
        \draw (0:0) node[minimum width=1.5pt,fill=black, label={180:$c(P)$}] {};
        \draw (-20:0.4) node[label={-20:$r_{1}$}] {};
        \draw (0:1.55) node[opacity=0, text opacity=1, label={20:$r_2 = f_{2}$}] {};
        \draw (0:1.5) node[fill=white] {};
        \draw (0:1.5) node[minimum width=2pt] {};
        \draw (-80:1.5) node[fill=white] {};
        \draw (-80:1.5) node[minimum width=2pt] {};
        \end{scope}

        \end{tikzpicture}
        \caption{\small The blue region represents the possible location of $f_1$. The red region represent the other possible locations of $f_1$ (based on the position of the visible robots $U(P)$).}\label{fig:possible location of f1}
    \end{figure}
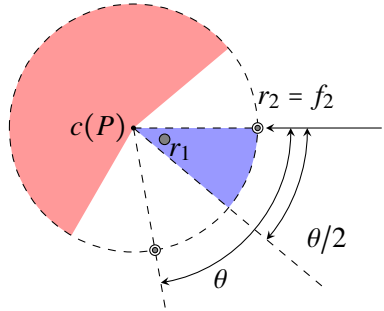
\end{proof}

\begingroup
\def\thelemma{\ref{lem:phase 3 terminates and C(P) is unchanged}}
\begin{lemma}[restated]
    After the \emph{Almost Pattern Formation} phase, each robot in $P'$ reaches its destination in $F'$, and the \emph{Termination} phase is executed.
\end{lemma}
\addtocounter{lemma}{-1}
\endgroup

\begin{proof}

After executing the three actions $cleanExterior(i)$, $locateEnoughRobots(i)$, and $removeRobotsInExcess(i)$ for a given $i$, we have $m_i$ robots on $C_i$ and $|interior(C_{i}) \cap P'| = |interior(C_{i-1}) \cap P'| - m_i = \sum_{j=i+1} ^{m}m_j$, so that we can execute the same procedure with $i+1$. If $i<m$ it is important to observe that some robots (those ordered to move in the last two cases) may still be in movement, but since they are now strictly between $C_i$ and $C_{i+1}$, they receive a new order with deterministic destination when executing the procedure with $i+1$. Hence, at the end of the procedure with $i = m$, all robots are static. 

After executing those actions for $i=1,2, \ldots, m$, each circle contains the proper number of robots \ie,  there are $m_i$ robots on $C_i$.

Then, we have to prove that when reaching their final destinations on their circles, the robots does not create a deadlock and does not change $C(P)$. First we show that there is no deadlock. Suppose we have on a circle $C$, $m$ robots $r_1 < r_2 < \ldots < r_m$ and $m$ destinations $d_1 < d_2 < \ldots <  d_m$. $r_i$ has destination $d_i$ and moves toward it (staying on $C$) in the direct orientation if $r_i < d_i$ and in the indirect orientation otherwise. For the sake of contradiction, suppose that $r_i < d_i$ and $r_i$ cannot reach $d_i$, even after an infinite number of activation. we observe that, to block $r_i$, $r_{i+1}$ must satisfy $r_i< r_{i+1} \leq d_1$ and $r_{i+1}$ is not able to free the way for $r_{i+1}$. This implies that $r_{i+1}\leq r_{i+3}\leq d_1$. Recursively, this means that $r_m \leq d_1$. But nothing blocks $r_m$ to reach $d_m$ when $r_m\leq d_m$, a contradiction.

Now suppose for the purpose of contradiction that $C(P)$ is modified. This means that there exist two robots $r$ and $r'$ on $C_1$ that form an angle greater than $\pi$. 
Before $C(P)$ is modified, they form an angle of at most $\pi$, so that one robot's movement on $C_1$ in the direct orientation, and the other's movement on $C_1$ in the indirect orientation. This is possible only if there is no point in $F$ on $C_1$ between $r$ and $r'$, which is a contradiction with the fact that $C(P)=C(F)$.
\end{proof}

\begingroup
\def\thelemma{\ref{lem:robot election}}
\begin{lemma}[restated] The following properties hold: \emph{(i)} eventually one robot is aware it is elected with probability one, and \emph{(ii)} once a robot is aware it is elected, another robot cannot be elected.
\end{lemma}
\addtocounter{lemma}{-1}
\endgroup
\begin{proof}
    \emph{i)} Suppose first that robots always reach their destination. 
    Initially the configuration is in a state where there are at least two robots whose location (or destination for those that are moving) are the closest to the center. 
    Let $r$ be the first (or one of the first) activated robot among them. With probability $1/2$, it chooses to move toward the center. Let $r'$ be another robot among them. If $r'$ is activated after $r$ begins its movement, $r'$ does not move, otherwise it moves away from the center with probability $1/2$. So with probability greater than $1/2^n$, \emph{regardless of the scheduler's choices}, $r$ is the only robot to move toward the center. After that, if another robot is activated before the next activation of $r$, it does not move. At the next activation of $r$, with probability $1/2$, $r$ chooses to move toward the center and becomes elected (and $r$ is aware it is elected when it is next activated since the other robots are static). So, with probability greater than $1/2^{n+1}$, $r$ is aware it is elected.
    If this does not happen, \ie, if the first or the second choice of $r$ is to move away the center or if another robot chose to move toward the center, then the configuration gets back to its initial state. So, since the probability of success is constant (for a given $n$), eventually a robot is aware it is elected with probability one.  
    Now, if robots do not always reach their destinations, the probability that $r$ moves by a distance $d$ toward the center (to become elected), is $1/2^{\left\lceil\frac{d}{\delta}\right\rceil}$, instead of $1/4$ in the first case. Indeed, instead of choosing two times to move toward the center, now $r$ needs to choose $\left\lceil\frac{d}{\delta}\right\rceil$ times to move toward the center. The probability that one of the other robot chooses to move away is still $1/2$. So that there is again a non null probability (greater than $1/2^{n+\left\lceil\frac{d}{\delta}\right\rceil}$) that a robot is elected, which implies that eventually a robot is aware it is elected with probability one.
    
    \emph{ii)} Once a robot $r_e$ is elected, the other robots are currently either moving away, not moving, or moving toward the center by a distance at most the eighth of their distance to the center. In each case, when another robot looks again (after it finishes its movement) it sees that the robot $r_e$ is the only closest robot to the center, and then it chooses not to move. Indeed, if another robot $r$ is moving toward the center, it is by a distance at most $|r|/8$. Since we have $|r_e| < \frac{7}{8}|r|$ when $r$ started its movement, we have $|r_e|<|r|$ after $r$ finishes it. After $r$'s movement, $r_e$ is maybe no longer elected, but since $r_e$ was aware it was elected, it already chooses to move on its circle to create a $1/8$-shifted-regular set (and it moves by a non null distance, so that in the next look phase, it is shifted).
\end{proof}

\begingroup
\def\thelemma{\ref{lem:phase FCB1 terminates}}
\begin{lemma}[restated]
    After executing the \emph{FCB1} phase, the configuration is static and associated with either the \emph{Termination} phase or the \emph{Almost Pattern Formation} (in finite time with probability one).
\end{lemma}
\addtocounter{lemma}{-1}
\endgroup

\begin{proof}
We proved in Lemma~\ref{lem:robot election} that a unique robot is elected with probability one. During the robot election the configuration still has the same CEB-set and we cannot have $|r_1| \leq |r_2|/2$ so that the configuration cannot be guided and remains associated with the FCB1 phase. Once a robot is elected, it is shifted. When a robot $r$ is shifted, no robot is located to a point (or has a destination) $p$ such that $|p| = |r|$, so that every other robot in the CEB-set receives a new destination when it sees that the configuration contains a CEB-set with a shifted angle $\theta/8$. The shifted robot increases the shift angle from $\theta/8$ to $\theta/4$ only when all the other robots in the CEB-set have reached their destinations. Once the shift angle is greater than $\theta/8$ all the robots are static. When the shifted robot has a shift angle of $\theta/4$, it moves toward the center and the other robots remain static. When the shifted robot $r$ stops \ie, when $|r| = \min(|r_2|, |f_2|)/2$, the configuration is guided and all the robots are static.
\end{proof}



\subsection{Proof of Theorem \ref{thm:sym conf have a CEB}}
The following algorithm presents the construction of the CEB-set when the configuration $P$ is regular or sym-regular (but not equiangular or biangular)
\begin{algorithm}[h]\small
    \SetAlgoRefName{constructCEBSet}
    \caption{}
    \label{algo:constructCEBSet}
    $Q = \emptyset$\quad $Ignore\leftarrow\emptyset$\\
    \While{$P\setminus (Ignore\cup Q)\neq\emptyset$}{
    Let $S$ be the set of smallest robots in $P\setminus (Ignore\cup Q)$ according to the partial ordering of robot\\
    \textbf{if} $Q\cup S$ holds $P$ \textbf{then} $Ignore\leftarrow Ignore\cup S$\\
    \textbf{else if }$Q\cup S$ is equiangular or biangular \textbf{then} $Q\leftarrow Q\cup S$\\
    \textbf{else} stop
    }
\end{algorithm}
\begingroup
\def\thetheorem{\ref{thm:sym conf have a CEB}}
\begin{theorem}[restated]
     For a configuration $P$, if $\rho(P)>1$, then $P$ has a unique, non-empty, CEB-set.
\end{theorem}
\addtocounter{property}{-1}
\endgroup
\begin{proof}
    If $\rho(P)>1$, then the configuration is regular or sym-regular. If $P$ is  biangular or equiangular, then $CEB(P) = P$. 
    Otherwise, if $P$ is regular, then the first set $S$ in the construction of $Q$ (algorithm \ref{algo:constructCEBSet}) does not hold $P$, so $S$ is added to $Q$, which becomes non-empty. If $P$ is sym-regular and not regular, then the set $S$ of smallest robots always contains one or two robots. Since $n\geq 5$, the while loop is executed at least three times. As we cannot have three disjoints subset of $P$ that hold $P$, $Q$ is not empty.
\end{proof}

\subsection{Uniqueness of the shifted robot}

Let $P$ be a $n$-robot configuration with $n\geq 5$. $\theta(P)$ is the minimum angle between two robots in $P$, centered at the center $c(P)$. We show that, for a shifted robot to be unique, it is sufficient that the shift angle is at most $\theta(P)/2$, and the orientation of the shift (the inverse of the orientation where we have to rotate the robot to create a CEB-set configuration) must reduce the minimum angle it forms with the other robots of the configuration. In the sequel, let $W(P)$ denote the position of the Weber point of set $P$.


\begin{lemma}\label{lem:property of the Weber point}
    Let $P$ be a $n$-robot configuration and $P' = P\setminus\{r\}\cup \{r'\}$. If $\overrightarrow{u} = \overrightarrow{W(P')W(P)}$, then for any point $p\in P\setminus\{r\}$, we have:
    \[
        \cos\left(\ang\left(\overrightarrow{u}, \overrightarrow{W(P)r}\right)\right) - \cos\left(\ang\left(\overrightarrow{u}, \overrightarrow{W(P')r'}\right)\right) \geq \cos\left(\ang\left(\overrightarrow{u}, \overrightarrow{W(P')p}\right)\right) - \cos\left(\ang\left(\overrightarrow{u}, \overrightarrow{W(P)p}\right)\right)
    \]
\end{lemma}
\begin{proof}
    It is known that $W(P)$ is such that
    \[
        \sum_{p\in P} \overrightarrow{v_p}=0
    \]
    where $\overrightarrow{v_p}$ is the normalized vector $\overrightarrow{v_p} = \frac{\overrightarrow{W(P)p}}{|W(P)-p|}$, or any vector such that $||v_p|| \leq 1$ if $p=W(P)$.
    The equality is true by taking only the $\overrightarrow{u}$ component of each vector, which implies
    \[
        \sum_{p\in P} \cos\left(\ang\left(\overrightarrow{u}, \overrightarrow{W(P)p}\right)\right)=0
    \]
    Where we assume that $\cos\left(\ang\left(\overrightarrow{u}, \overrightarrow{W(P)p}\right)\right)$ can be any number in $[-1,1]$ if $p=W(P)$. The same equality holds in $P'$. By subtraction, we get:
    \begin{equation}\label{eq:equality between sum of consines}
        \cos\left(\ang\left(\overrightarrow{u}, \overrightarrow{W(P)r}\right)\right) - \cos\left(\ang\left(\overrightarrow{u}, \overrightarrow{W(P')r'}\right)\right) = \sum_{p\in P\setminus\{r\}} \cos\left(\ang\left(\overrightarrow{u}, \overrightarrow{W(P')p}\right)\right) - \cos\left(\ang\left(\overrightarrow{u}, \overrightarrow{W(P)p}\right)\right)
    \end{equation}
    One can observe that, for every point $p\in P\setminus\{r\}$, we have:
    \[
        \cos\left(\ang\left(\overrightarrow{u}, \overrightarrow{W(P')p}\right)\right) \geq 
        \cos\left(\ang\left(\overrightarrow{u}, \overrightarrow{W(P)p}\right)\right)
    \]
    so that each term in the sum of the right hand side of equation (\ref{eq:equality between sum of consines}) is positive, which implies that the left hand side is not smaller than each term in the right hand side, Q.E.D.
    
\end{proof}
\begin{lemma}\label{lem:angle of robots are less that angle of the moving robot}
    Let $P$ be a $n$-robot configuration and $P' = P\setminus\{r\}\cup \{r'\}$, with $\theta = \ang_{\min}(r,W(P),r') \leq \theta(P)$, and $|r'|_{W(P)} = |r|_{W(P)} = \min_{p\in P}|p|_{W(P)}$.
    We have $\forall p\in P$, $ang_{\min}(W(P), p, W(P')) < \theta$.
\end{lemma}
\begin{proof}
\tolerance=1000
For this lemma, we assume that all angles are computed using a global orientation such that ${\ang_{\min}(r,W(P),r') = \ang(r,W(P),r')}$. Also we define $\overrightarrow{u} = \overrightarrow{W(P')W(P)}$ and $\overrightarrow{\arg}_p(r) = \arg(\overrightarrow{u}, \overrightarrow{pr})$

    From the previous lemma, we have:
    \begin{equation}\label{eq:difference in cosinus of the Weber Point full notation}
        \cos(\overrightarrow{\ang}_{W(P)}(r)) - \cos(\overrightarrow{\ang}_{W(P')}(r')) \geq \cos(\overrightarrow{\ang}_{W(P')}(p)) - \cos(\overrightarrow{\ang}_{W(P)}(p)) 
    \end{equation}
    With the notations of Figure~\ref{fig:proof angle of robot, configuration}:
    \begin{equation}\label{eq:difference in cosinus of the Weber Point}
        \cos(\tau) - \cos(\tau+\theta - \gamma) \geq \cos(\beta') - \cos(\beta' + \gamma') 
    \end{equation}

    We observe that $\pi \leq \overrightarrow{arg}_{W(P)}(r) < \overrightarrow{arg}_{W(P')}(r')$ is not possible, because this would mean that the half-lines $HL(W(P), r)$ and $HL(W(P'), r')$ intersect, which is not possible.
    
    We can have three cases:\\
    \noindent\textbf{The case $\overrightarrow{arg}_{W(P)}(r) \leq \pi < \overrightarrow{arg}_{W(P')}(r') $:} then 
    \[
        cos(\overrightarrow{arg}_{W(P)}(r)) - cos(\overrightarrow{arg}_{W(P')}(r')) \leq cos(\overrightarrow{arg}_{W(P)}(r)) + 1 < 1 - cos(\theta)
    \]
    On the other hand, if $arg_{\min}(W(P), p,W(P')) \geq \theta$, we have:
    \[
        \cos(\overrightarrow{arg}_{W(P')}(p)) - cos(\overrightarrow{arg}_{W(P)}(p) \geq 1 - cos(arg_{\min}(W(P), p,W(P'))) \geq 1 - cos(\theta)
    \]
    which contradicts Equation~\ref{eq:difference in cosinus of the Weber Point full notation}, then we have $arg_{\min}(W(P), p,W(P')) < \theta$.
    
    \noindent\textbf{The case $\overrightarrow{arg}_{W(P')}(r') < \pi < \overrightarrow{arg}_{W(P)}(r)$:} again, we have 
    \[
        cos(\overrightarrow{arg}_{W(P)}(r)) - cos(\overrightarrow{arg}_{W(P')}(r')) \leq 1 - cos(\overrightarrow{arg}_{W(P')}(r')) < 1 - cos(\theta)
    \]
    with the same argument, this implies that $arg_{\min}(W(P), p,W(P')) < \theta$.
    
    \begin{figure}\centering
        \includegraphics[width=10cm]{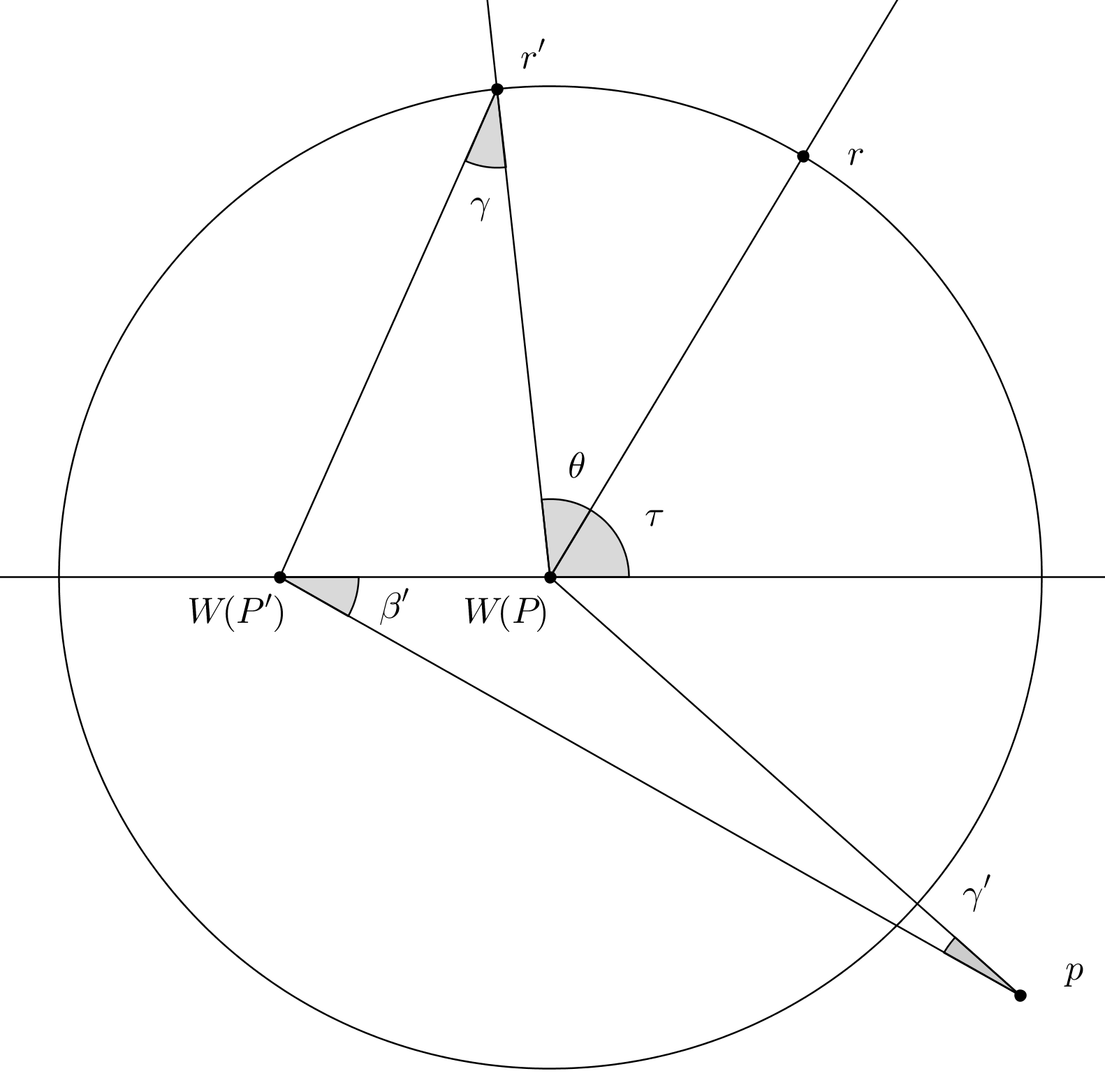}
        \caption{Part of the configuration $P$, for the proof that $\gamma'$ cannot be greater than $\theta$.}\label{fig:proof angle of robot, configuration}
    \end{figure}
    \noindent\textbf{The case $\overrightarrow{arg}_{W(P)}(r) < \overrightarrow{arg}_{W(P')}(r') \leq \pi $:}
    This case corresponds to the configuration shown in Figure~\ref{fig:proof angle of robot, configuration}. Using the notations of the figure, we want to prove that $\gamma' < \theta$ for any position of $r$, $r'$, $W(P)$, $W(P')$, and $p$ (for an arbitrary $p\in P\setminus\{r\}$) verifying the conditions of the lemma. Like in previous cases, we suppose that $\gamma' \geq \theta$ and we reach a contradiction.

For simplicity, we suppose that $|W(P) - r| = 1$ and we characterize the configuration using the angles $\tau$, $\gamma$, $\theta$, $\tau'$ and the distance $l = |W(P) - p|$ rather that the positions of the points. To this purpose, we start by giving the distance $\omega = |W(P) - W(P')|$ with respect to the angles $\tau$, $\gamma$, and $\theta$ using the law of sinus:
\[\omega(\tau, \gamma, \theta) = \frac{\sin(\gamma)}{\sin(\tau + \theta - \gamma)}\]
Then, the angle $\beta'$ is computed using the law of sinus:
\[
\beta'(\tau, \gamma, \theta, l, \gamma') = \beta_1'(\tau, \gamma, \theta, l, \gamma') =\arcsin\left(\frac{l\sin(\gamma')}{\omega(\tau, \gamma, \theta)}\right)\,\text{ or }\,
\beta'(\tau, \gamma, \theta, l, \gamma') = \beta_2'(\tau, \gamma, \theta, l, \gamma') =\pi - \beta_1'(\tau, \gamma, \theta, \gamma', l)
\] 
which exists only if $\omega(\tau, \gamma, \theta) \geq l\sin(\gamma')$.
Then, we define $f_6$:
\[f_6: \left(\tau, \gamma, \theta, l, \gamma', i\right) \mapsto cos\left(\beta_i'\left(\tau, \gamma, \theta, l, \gamma'\right)\right) - cos\left(\beta_i'\left(\tau, \gamma, \theta, l, \gamma'\right) + \gamma'\right) - \left(cos\left(\tau\right) - cos\left(\tau + \theta - \gamma\right)\right)\]
The domain of definition of $f_6$ is such that $\theta\in[0, \pi/4]$, $\gamma\in[0, \theta)$, $\tau\in[0,\pi - \theta]$, $l\geq 1$, $\gamma'\in[0, \pi/2]$, $i\in\{1,2\}$, and $\omega(\tau, \gamma, \theta) \geq l\sin(\gamma')$. In the sequel, except stated otherwise, $f_6(\tau, \gamma, \theta, \gamma', l, i)$ denotes the image by $f_6$ of an arbitrary tuple $(\tau, \gamma, \theta, \gamma', l, i)$ in its definition domain.
We want to show that if $\gamma' \geq \theta$,  then $f_6(\tau, \gamma, \theta, \gamma', l, i) > 0$, which contradicts Equation \ref{eq:difference in cosinus of the Weber Point}. So now, for the sake of contradiction we suppose that $\gamma' \geq \theta$. The following claims 1 and 3 prove the lemma.

\noindent\textbf{Claim 1:} if $\omega(\tau, \gamma, \theta) > 1$, then $f_6(\tau, \gamma, \theta, l, \gamma', i) > 0$.
First, if $\omega(\tau, \gamma, \theta) > 1$, we can use a simple lower bound:
\[
\cos(\beta_1'(\tau, \gamma, \theta, l, \gamma')  - \gamma') - \cos(\beta_1'(\tau, \gamma, \theta, l, \gamma')) \geq 1 -\cos(\gamma') > 1 -\cos(\theta) 
\]
Also, $\omega(\tau, \gamma, \theta) > 1$ implies that $\sin(\gamma) > \sin(\tau+\theta-\gamma)$, which in turn implies that $\gamma > \tau+\theta-\gamma$. Then we have:
\[
    \theta > \gamma > \tau+\theta-\gamma > \tau
\]
so that 
\[ 
    1 - \cos(\theta)\geq\cos(\tau) - \cos(\tau+\theta-\gamma) \Rightarrow f_6 (\tau, \gamma, \theta, l, \gamma', i)> 0 
\]

\noindent\textbf{Claim 2:} if $\omega(\tau, \gamma, \theta) \leq 1$, The minimum of $f_6$ is the minimum of $f_3: (\tau, \gamma, \theta) \mapsto f_6(\tau, \gamma, \theta, 1, \theta, 2)$.
We want to show that the minimum of $f_6$ is reached when $\gamma'$ is minimum (\ie, when $\gamma' = \theta$), when $l=1$, and with $i=2$.

Since $l > \omega(\tau, \gamma, \theta)$, then $\beta_1' = \beta_1'(\tau, \gamma, \theta, l, \gamma') \geq \gamma'$. First, we have:
\[
    \cos(\beta_2') - \cos(\beta_2' + \gamma') = \cos(\pi - \beta_1') - \cos(\pi - \beta_1' + \gamma') = \cos(\beta_1' - \gamma') - \cos(\beta_1')
\]
Since $\beta_1' \leq \pi/2$, from the concavity of the cosine in $[0, \pi/2]$ and, if  $\beta_1' + \gamma' > \pi/2$, from the symmetry of the cosine with respect to $\pi/2$, we have: 
\[
    \cos(\beta_1' - \gamma') - \cos(\beta_1') \leq \cos(\beta_1') - \cos(\beta_1' + \gamma')
\]
So that the minimum of  $f_6$ is obtained with $i=2$. 
Moreover, since
\[
\pi/2\geq\beta_1' \geq \beta'(\tau, \gamma, \theta, l, \theta)
\]
Then, again, by concavity of the cosine we have:
\begin{align*}
\cos(\beta_2') - \cos(\beta_2' + \gamma') = \cos(\beta_1' - \gamma') - \cos(\beta_1') &\geq \cos(\beta'(\tau, \gamma, \theta, l, \theta)  - \gamma') - \cos(\beta'(\tau, \gamma, \theta, l, \theta))\\
&\geq \cos(\beta'(\tau, \gamma, \theta, l, \theta)  - \theta) - \cos(\beta'(\tau, \gamma, \theta, l, \theta))
\end{align*}
Finally, again, by concavity of the cosine, since $\beta'(\tau, \gamma, \theta, l, \theta) \leq \beta'(\tau, \gamma, \theta, 1, \theta)$, we have:
\[
    \cos(\beta'(\tau, \gamma, \theta, l, \theta)  - \theta) - \cos(\beta'(\tau, \gamma, \theta, l, \theta)) \geq \cos(\beta'(\tau, \gamma, \theta, 1, \theta)  - \theta) - \cos(\beta'(\tau, \gamma, \theta, 1, \theta))
\]
So that $f_6(\tau, \gamma, \theta, l, \gamma', i) \geq f_6(\tau, \gamma, \theta, 1, \theta, 2) = f_3(\tau, \gamma, \theta)$. 

\noindent\textbf{Claim 3:} if  $\omega(\tau, \gamma, \theta) \leq 1$, then $f_3(\tau, \gamma, \theta) > 0$.
Let $f_{\tau} : \tau \mapsto f_3(\tau, \gamma, \theta)$ for fixed $\gamma$ and $\theta$. By analyzing Function $f_{\tau}$, we observe that $f_{\tau}$ is increasing and then decreasing, so that the minimum of $f_{\tau}$ is obtained with the greatest or smallest possible value of $\tau$. For fixed $\theta \leq\pi/4$, and $\gamma \leq \theta$, we must have $\tau \leq \pi - \theta$ and $\tau \geq 2\gamma - \theta$ (because $\omega(\tau, \gamma, \theta) \leq 1$).

On the one hand, we have (where $\gamma$ and $\theta$ are chosen in the definition domain of $f_3$):
\[
    \forall\gamma, \forall\theta,\quad f_3(\pi - \theta, \gamma, \theta) = 1 - cos(\theta) + cos(\theta) - cos(\gamma) > 0
\]
and on the other hand, we have:
\[
    \forall\gamma,\forall \theta,\quad f_3(2\gamma - \theta, \gamma, \theta) = 1 - cos(\theta) - cos(2\gamma - \theta) + cos(\gamma) > 0
\]
    Where the last inequality is true since $cos(\gamma) - cos(\theta) > 0$.

\end{proof}

\begin{lemma}\label{lem:the ordering is invariant by shift} 
Let $P$ be a $n$-robot configuration. Suppose the robots are indexed in the clockwise order $r_1$, $r_2$, $\ldots$, $r_n$ around the Weber point $W(P)$. Let $P' = P\setminus\{r\}\cup \{r'\}$, with $ang_{\min}(r,W(P),r') \leq \theta(P)/2$ and $|r'|_{W(P)} = |r|_{W(P)} = \min_{i\in[1,n]}|r_i|_{W(P)}$.
    Then, the robots in $P'$ are ordered in the same way as in $P$ around $W(P')$ (with $r'$ instead of $r$).
\end{lemma}

\begin{proof}
Using the same notation as the previous lemma, we define $\overrightarrow{u} = \overrightarrow{W(P')W(P)}$ and $\overrightarrow{\arg}_p(r) = \arg(\overrightarrow{u}, \overrightarrow{pr})$.

We suppose that the ordering of the robots in $P$ is such that $\overrightarrow{\arg}_{W(P)}(r_i) < \overrightarrow{\arg}_{W(P)}(r_{i+1})$, for all $i\in[1..n-1]$ (a circular permutation of the ordering can make this true).

For all $i\in[1..n]$, let $r_i' = r_i$ if $r\neq r_i$, otherwise let $r_i'= r'$. We now show that $r_1',\ldots, r_n'$ is an ordering  of the robots in $P'$ around $W(P')$. If $r\neq r_{i}$ and $r\neq r_{i+1}$, we have:
\begin{align*}
    \overrightarrow{\arg}_{W(P')}(r_{i}') &< \overrightarrow{\arg}_{W(P)}(r_{i}) + \theta(P)/2
    & \text{from Lemma~\ref{lem:angle of robots are less that angle of the moving robot}} \\
    &\leq \overrightarrow{\arg}_{W(P)}(r_{i+1}) - \theta(P)/2
    & \text{by definition of $\theta(P)$} \\
    &\leq \overrightarrow{\arg}_{W(P')}(r_{i+1}')
    & \text{from Lemma~\ref{lem:angle of robots are less that angle of the moving robot}} 
\end{align*}
If $r= r_{i+1}$, resp. $r= r_{i}$, then we have $\overrightarrow{\arg}_{W(P)}(r) < \overrightarrow{\arg}_{W(P')}(r')$, resp.  $\overrightarrow{\arg}_{W(P')}(r') < \overrightarrow{\arg}_{W(P)}(r) +\theta(P)$, so that the previous inequality holds.

Overall, $\overrightarrow{\arg}_{W(P')}(r_i') < \overrightarrow{\arg}_{W(P')}(r_{i+1}')$ for all $i\in[1..n-1]$, and the ordering is preserved.

\end{proof}


\begin{theorem}\label{thm:n-robots n-regular, the elected robot is unique}
    Let $P$ be a $n$-robot ($n\geq 5$) biangular configuration that contains a shifted robot. The shifted robot is unique.
\end{theorem}

\begin{proof}
Let $P$ be a biangular configuration with a shifted robot $r$ (resp., $r'$), associated to the biangular set $P_{bi}=P\setminus\{r\}\cup\{r_{bi}\}$ (resp., associated to the biangular set $P_{bi}'=P\setminus\{r'\}\cup\{r_{bi}'\}$). For the sake of contradiction, we suppose that $r\neq r'$. From Lemma~\ref{lem:the ordering is invariant by shift}, the ordering is unchanged between $P$, $P_{bi}$ and $P_{bi}'$. So, let $r_1$, $\ldots$, $r_n$ be any ordering of robots in $P$ around $W(P)$. An ordering of robots in $P_{bi}$ (resp., $P_{bi}'$) is obtained by replacing $r=r_{i_0}$ by $r_{bi}$ (resp., $r'=r_{i_0'}$ by $r_{bi}'$). 

If $n\geq 5$ and $P$ is equiangular, there are $3$ robots in $P_{bi}\cap P_{bi}'$ ordered in the same way around the center of equiangularity. Their angles are uniquely determined by the difference between their indexes, indeed we have:
\[
    \ang(r_i, W(P), r_j)= \ang(r_i, W(P'), r_j) = \frac{|i-j|2\pi}{n}
\]
Let $arc_{\alpha}(a,b)$ be the set of points $p$ such that $\ang(a,p,c) = \alpha$. $arc_{\alpha}(a,b)$ is a circular arc from $a$ to $b$. Two circular arcs intersect in at most two points, so that there is at most one point $p\notin \{a,b,c\}$ in $arc_{\alpha}(a,b)\cap arc_{\beta}(b,c)$. This implies that that three robots are enough to deduce the position of the unique possible center of equiangularity of $P_{bi}$ and $P_{bi}'$. With a given center, there cannot be two shifted robots. Indeed, if $c(P_{bi}) = c(P_{bi}')$, then the angles formed with $r'$ and the other robots are the same as the robot $r_{bi}'$ and $|r_{bi}'|_{c(P_{bi}')} = |r'|_{c(P_{bi}')}$ \ie, $r' = r_{bi}$.

If $n\geq 8$ and $P$ is biangular, there are $3$ robots in $P_{bi}\cap P_{bi}'$ whose difference between indexes are even, and the previous argument holds.

The last case to consider is when $n=6$ and $P$ is biangular. In this case, if there are $3$ robots in $P_{bi}\cap P_{bi}'$ whose difference between indexes are even, then the previous argument holds. Otherwise, without loss of generality, either $\{r_1,r_2,r_3,r_4\}\in P_{bi}\cap P_{bi}'$ or $\{r_1,r_3,r_4,r_6\}\in P_{bi}\cap P_{bi}'$.

\tolerance=1000
In the first case, the center of biangularity is the unique point of intersection of $arc_{2\pi/3}(r_1, r_3)$ with ${arc_{2\pi/3}(r_2, r_4)}$. Indeed, the intersection is unique due to the ordering of robots.

In the second case, there can be two intersection points between $arc_{2\pi/3}(r_1, r_3)$ and   $arc_{2\pi/3}(r_4, r_6)$. We suppose $r=r_2$ and $r' = r_5$. 
Let $\alpha\leq\beta$, resp. $\alpha'\leq\beta'$, the two angles in the string of angles $SA_{c(P_{bi})}(P_{bi})$, resp.  $SA_{c(P_{bi}')}(P_{bi}')$. 
For simplicity, we can suppose that $\alpha = \arg(r_1, c(P_{bi}), r_{bi})$. 

On the one hand, the robot $r$ must be shifted in the orientation that decreases its angle with respect to the other robots \ie, decreasing its angle with $r_1$. 
On the other hand, $r$ is in $P_{bi}'$ so that $r$ must verify $\arg(r_6, c(P_{bi}'), r) = 2\pi/3$. In particular, this implies that $\arg(r_1, c(P_{bi}'), r) < \arg(r_1, c(P_{bi}), r_{bi})$
so that $\alpha' < \alpha$.
Applying the same argument with $r'$ leads to the contradiction $\alpha<\alpha' < \alpha$.

\end{proof}

\begingroup
\def\thetheorem{\ref{thm:uniqueness of the regular set}}
\begin{theorem}[restated]
    Let $n\geq5$, and $P$ be a $n$-robot configuration that contains a CEB-set with a shifted robot. Then, the shifted robot is unique.
\end{theorem}
\addtocounter{theorem}{-1}
\endgroup

\begin{proof}
    Let $P$ be a $n$-robot configuration that contains a CEB-set $Q$ with a shifted robot.
    If the whole configuration is biangular, Theorem~\ref{thm:n-robots n-regular, the elected robot is unique} implies the result.
    
    Otherwise, the configuration is regular or sym-regular, the center is the center of $SEC(P)$, and, even if different robots can be  considered as shifted (when $|Q|<5$), there can be only one robot that minimizes the angle with the other robots.
\end{proof}

\end{document}